\documentclass{article}
\usepackage{graphicx}
\usepackage[space]{grffile}
\usepackage{latexsym}
\usepackage{textcomp}
\usepackage{longtable}
\usepackage{tabulary}
\usepackage{booktabs,array,multirow}
\usepackage{amsfonts,amsmath,amssymb}
\usepackage{natbib}
\usepackage{url}
\usepackage{hyperref}
\hypersetup{colorlinks=false,pdfborder={0 0 0}}
\usepackage{etoolbox}
\usepackage{amssymb}
\usepackage{amsmath}
\usepackage{graphicx}
\usepackage{graphics}
\usepackage{wrapfig}\label{key}
\usepackage{authblk}
\usepackage[all]{xy}
\usepackage{bbm}
\usepackage{float}
\usepackage{tikz}
\usetikzlibrary{arrows,calc,shapes,decorations.pathreplacing}
\usepackage{wrapfig}
\usepackage{sidecap}
\usepackage[T1]{fontenc}
\usepackage[abs]{overpic}
\usepackage{color, colortbl}
\usepackage{soul}

\newtheorem{theorem}{Theorem}[section]

\newtheorem{proposition}[theorem]{Proposition}

\newtheorem{remark}[theorem]{Remark}

\newenvironment{proof}[1][Proof]{\begin{trivlist}
\item[\hskip \labelsep {\bfseries #1}]}{\end{trivlist}}

\def\E{\mbox{\rm E}}
\def\P{\mbox{P}}
\def\var{\mbox{Var}}
\def\gamma{\mbox{Gamma}}

\newcommand{\qed}{\nobreak \ifvmode \relax \else
      \ifdim\lastskip<1.5em \hskip-\lastskip
      \hskip1.5em plus0em minus0.5em \fi \nobreak
      \vrule height0.75em width0.5em depth0.25em\fi}

\newif\iflatexml\latexmlfalse

\AtBeginDocument{\DeclareGraphicsExtensions{.pdf,.PDF,.eps,.EPS,.png,.PNG,.tif,.TIF,.jpg,.JPG,.jpeg,.JPEG}}

\usepackage[utf8]{inputenc}
\usepackage[english]{babel}

\usepackage{siunitx}

\usepackage{amsmath}
\usepackage{wrapfig}
\usepackage{authblk}
\usepackage[all]{xy}
\usepackage{bbm}
\usepackage{float}
\usepackage{tikz}
\usetikzlibrary{arrows,calc,shapes,decorations.pathreplacing}
\usepackage{wrapfig}
\usepackage{sidecap}
\usepackage{caption}
\usepackage{url}
\usepackage[abs]{overpic}
\usepackage[toc,page]{appendix}
\usepackage{soul}
\usepackage{setspace}
\usepackage{kantlipsum}

 \newcommand{\R}{R}

\def\E{\mbox{\rm E}}
\def\P{\mbox{P}}
\def\var{\mbox{Var}}

\iflatexml

\else

\fi

\title{Spatial modeling of randomly acquired characteristics on outsoles with application to forensic shoeprint analysis}

\author{Naomi Kaplan-Damary\footnote{Department of Statistics, The Hebrew University of Jerusalem, 91905 Jerusalem, Israel. naomi.kaplan@mail.huji.ac.il}, Micha Mandel\footnote{Department of Statistics, The Hebrew University of Jerusalem, 91905 Jerusalem, Israel. micha.mandel@mail.huji.ac.il }, Yoram Yekutieli \footnote{Hadassah Academic College
		Dept. of Computer Science
		37 Hanevi'im St. P.O.Box 1114,
		9101001 Jerusalem, Israel.
		yoramye@hac.ac.il}, Sarena Wiesner\footnote{Israel National Police
		Division of Identification and Forensic Science (DIFS)
		1 Bar-Lev Road
		91906 Jerusalem, Israel. sarenawiz@gmail.com}, Yaron Shor\footnote{Israel National Police Division of Identification and Forensic Science (DIFS) 1 Bar-Lev Road
		91906 Jerusalem, Israel. yaronshor@gmail.com}.}

\begin{document}

\date{}
\maketitle

\selectlanguage{english}
\begin{abstract}
Footwear comparison is used to link between a suspect's shoe and a footprint found at a crime scene. Investigators compare the two items using randomly acquired characteristics (RACs), such as scratches or holes. However, to date, the distribution of RAC characteristics has not been investigated thoroughly, and the evidential value of RACs is yet to be explored. An important question concerns the distribution of the location of RACs on shoe soles, which can serve as a benchmark for comparison. The location of RACs is modeled here as a point process over the shoe sole and a data set of 386 independent shoes is used  to estimate its rate function. The analysis is somewhat complicated as the shoes are differentiated by shape, level of wear and tear and contact surface. This paper presents methods that take into account these challenges, either by using natural cubic splines on high resolution data, or by using a piecewise-constant model on larger regions defined by experts' knowledge. It is shown that RACs are likely to appear at certain locations, corresponding to the foot’s morphology. The results can guide investigators in determining the evidential value of footprint comparison.

\end{abstract}%

\section{Introduction}
In recent years, forensic methods have been criticized for their shortcomings in providing courts with objective and quantitative answers to the question of whether a sample from a suspect matches a sample found at the crime scene. Unlike DNA that is used routinely to link suspects to crime scenes because of its scientific objectivity and accessible documentation, the evaluation of other types of evidence such as shoeprints, hair, and even fingerprints has not reached this gold standard. Both the 2009 NRC report, ``Strengthening Forensic Science in the United States: A Path Forward,'' and the 2016 PCAST report to President Obama, ``Forensic Science in Criminal Courts: Ensuring Scientific Validity of Feature-Comparison Methods,'' have called for the strengthening of the scientific basis of forensic procedures (NRC, 2009; PCAST, 2016).

Here shoeprint comparison is considered. The identification of footwear impressions is based on the comparison of a print found at the crime scene with a  print made from a suspect's shoe (Bodziak, 1999); see 
Supporting web materials~1 and a short film (Kaplan-Damary, 2014) for a detailed description of the process of comparing shoeprints. The analysis of shoeprints by experts is done in two broad stages. First, the pattern, size and wear of the shoe sole are compared to the crime scene print. If these do not fit, the analysis is stopped and the pair is classified as a non-match. 
In the second stage, the forensic expert examines whether randomly acquired characteristics (RACs) on the shoe sole match the RACs on the print from the crime scene. These RACs have various characteristics, such as location, shape and orientation, that are used in the expert's evaluation.

\begin{figure}[tb]
	\begin{center}
		\includegraphics[width=0.5\columnwidth]{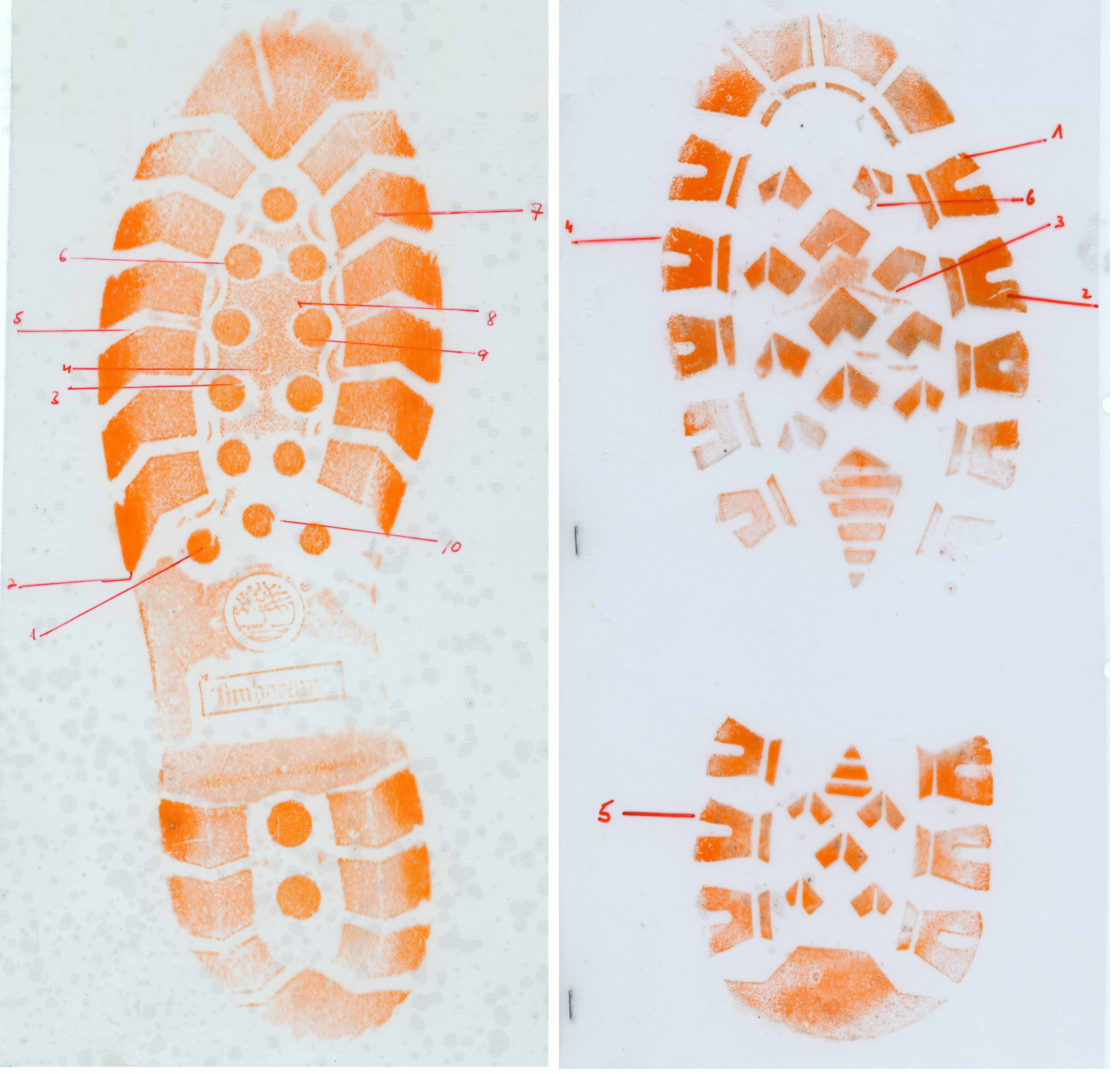}
	\end{center}
	\caption{Locations of RACs marked on lab prints of suspects' shoes (Courtesy of the Israel National Police Division of Identification and Forensic Science -- DIFS)}
	\label{fig:mark_RAC}
\end{figure}

The current study deals only with the location of RACs. Figure \ref{fig:mark_RAC} presents two lab prints taken from suspects' shoes, with the location of RACs marked by the examiner. The rarity of this set of RACs is of major interest, especially in establishing a link between suspect's shoes and the crime scene prints. Thus, the main focus of this study is understanding the spatial distribution of RACs' locations, and specifically whether they are distributed uniformly as assumed by Wiesner et al. (2019) and Stone (2006) or concentrated in certain areas. This is an essential step in evaluating the degree of rarity of a given set of RACs, i.e. the probability that a random shoeprint has a pattern of RACs that is sufficiently similar. Marks at sparsely populated locations would be of much greater value in determining a match than marks at highly populated locations
(Richetellia, 2019; Spencer and Murray, 2019).

Until recently, there were not enough data to reliably answer questions regarding the probability of RAC locations. Stone (2006) assumes a uniform distribution over the shoe sole as a simple theoretical model without any data to support it. Wiesner et al. (2019)  suggest a uniform model for the probability of location using a naive approach described in detail in Section 4. Richetelli, Bodziak and Speir (2019) use a large collection of about 1300 shoes and calculate the probability that RAC's in two independent shoes occur at the same location. While a simple adjustment is made to the difference in the contact surfaces in different shoes, the method neither exploits information from neighbor locations, nor produces an estimate for the rate function. Spencer and Murray (2019) use the database of  Wiesner et al. (2019)  and fit a quite complicated spatial hierarchical  Bayesian model to the location of RACs, allowing for shoe-specific rate functions. However, they use a local dependence model, which may be inappropriate when the location variable is subject to measurement errors.

Similar to Spencer and Murray (2019), here the RACs are modeled as a spatial process, and the intensity function is estimated in order to calculate the probability of observing RACs in different locations. We first study  a random effects continuous model for the rate function,  and then compare it to a simpler step-wise constant model that is less sensitive to measurement errors. This is done by using the Israeli Police Division of Identification and Forensic Science database of RACs, which includes 13,000 RACs from 386 lab shoeprints (Wiesner et al., 2019).

Estimation of intensity functions arises naturally in the framework of spatial statistics (Cressie, 1993), but there are several complications in the current analysis of shoe impressions.
First, spatial statistics typically deals with a single process that has many events, while in our case, there are many independent processes (shoe impressions), and each has only a small number of events (RACs); about 34 per shoe on average.
A second complication is the variability of shoes: they differ in their types and sizes. Moreover, different shoes have different contact surfaces i.e, the part of the sole that actually touches the floor or ground (see Figure~\ref{fig:mark_RAC}, which presents an example of lab prints. The area in orange is the contact surface). This fact limits the area in which RACs appear, thus affecting the probability of observing them.
On top of these difficulties, some shoes are scarred by many RACs while others have relatively few, apparently due to the level of wear and tear. This article models and estimates the intensity function while taking into account the challenges noted above.

The data used in this article are described in Sections 2 and the model is presented in Section 3. In Section 4, three estimators of the intensity function are introduced; a naive estimator, an estimator based on a random effects model and an estimator based on conditional maximum likelihood. These are presented in a pixel-based model which employs 
a logistic regression with natural cubic splines for the location variable, yielding a smooth intensity function. In addition, sub-sampling techniques are described in order to deal with computational challenges when using such models. 
 Using the exact location of RACs is problematic in practice due to characteristics of the data presented in Section 5. Instead, the shoe is divided into larger areas to which a piece-wise constant intensity function is fitted. Section 6 presents simulation results for the comparison of sub-sampling case-control techniques as well as a comparison of the piecewise constant estimators in different settings. Section 7 concludes the paper with a discussion.


\section{Data}\label{data}

The Israeli Police Division of Identification and Forensic Science (DIFS) has amassed one of the most comprehensive RAC databases, including some 386 lab prints and 13,000 RACs (Wiesner et al., 2019). An important initial preprocessing step was to normalize all shoe impressions to a standardized X-Y axis with identical length and orientation.
This was done by first marking a shoe-aligned coordinate system on each print and then standardizing the shoe according to this system: for each lab print, the top and bottom of the shoeprint were marked to indicate the direction of the major axis and to determine the length of the shoe. The axes' origin was set at the middle point between the two marked extremities. The minor axis was defined as the line perpendicular to the major axis that passes through the origin of axes. The standardization was done by transforming all measurements from image coordinates to the shoe aligned coordinate system as follows (for more details, see Wiesner et al., 2019):
\begin{enumerate}
	\item[a)] Translation of the marked origin of axes to (0,0).
	\item[b)] Rotation by the direction of the shoe aligned coordinate system.
 \item[c)] Scaling by the length of the shoeprint.
  \item[d)] Multiplying the x-value of the points (x being the horizontal axis) by -1 or 1, to mirror if needed such that all shoeprints will be turned to left shoes.
\end{enumerate}
RACs are assumed to have a two-dimensional shape. This study focuses on the location, measured as a point $(x,y)\in \mathbb{R}^2$, which is the center of gravity calculated as the mean of all pixels included in the RAC.

The number of RACs per shoe varies between 1 and 190 with an average of 34 (see 
Figure~1 in Supporting web materials~2)
 except for one shoe that has an unusually high number of RACs (309). The RACs were marked by different examiners who were supervised by forensic experts. RACs can be observed only on the contact surface, a feature which varies from one shoe to the other. This should be taken into account in the analysis as described in Section~\ref{model}. 
 The number of pixels with contact surface per shoe varies between 3631 and 19199 (see Figure~2 in Supporting web materials~2). It is also shown 
 that the pad of the shoe and the 4 circles at the heal more frequently contain a contact surface (see Figure~3 in Supporting web materials~2 for the cumulative contact surface of all shoes). There is a weak correlation between the number of pixels with contact surface per shoe and the number of RACs per shoe -- the Spearman correlation coefficient is equal to 0.116; see Figure~4 in Supporting web materials~2.


\section{Model} \label{model}
Consider $m$ independent shoes with different levels of wear and tear. It is assumed that the locations of RACs on each shoe follow a non-homogeneous Poisson process. Specifically, let $D\subset \mathbb{R}^2$ be the region representing the surface of a generic shoe, and let $B_i\subset D$ be the contact surface of shoe $i$, $i=1,\dots,m$. For any $A\subset D$, denote by $N_i(A)$ the number of RACs appearing on subset $A$ of shoe $i$. Note that $N_i (A)$ represents all of the RACs: those on the contact surface and those that are not, but only RACs on the contact surface are observed. $N_i$ is assumed to be a non-homogeneous Poisson point process on shoe $i$ with intensity function $\lambda^{(i)}$ and corresponding cumulative intensity function $\Lambda^{(i)}(A)=\int_{(x,y)\in A} \lambda^{(i)}(x,y)\,dx\,dy$. Let $a_i \in \mathbb{R}^+$ be a random variable that indicates the degree of wear and tear of shoe $i$, such that $a_i$, $i=1\dots,m$, are $iid$ with $\E(a_i)=1$ and $\var(a_i)=\sigma^2$. Our basic assumption is that
\begin{equation}\label{assump1}
\lambda^{(i)}(x,y)=\lambda^{(0)}(x,y)\cdot a_i, ~~(x,y)\in D,
\end{equation}
where $\lambda^{(0)}$ is joint to all shoes.
Thus, it is assumed that all shoes have the same shape of the intensity function and only differ by the shoe-specific parameter $a_i$ that determines the height of the function. An equivalent assumption to \eqref{assump1} is $\Lambda^{(i)}(A)=\Lambda^{(0)}(A)\cdot a_i, ~~A\subseteq D$, where  $\Lambda^{(0)}(A)=\int_{(x,y)\in A} \lambda^{(0)}(x,y)\,dx\,dy$.

The main goal is to estimate the baseline intensity function $\lambda^{(0)}$, which is, in general, a continuous function. To simplify the analysis while preserving the sole type for the observer, images were reduced in resolution from $7000 \times 3000$  to $397\times 307$ pixels.  Thus, the modeling assumption is that the function is constant within pixels in a high resolution grid. In Section \ref{est_large} we consider a partition into larger areas. In general, region $D$ is partitioned into $J$ subsets $A_1,\ldots,A_J$ such that $\cup _j A_j= D$ and $A_j \cap A_k=\phi$,   $\forall j\neq k$. The assumption is that $\lambda^{(0)}(x,y)=\lambda_j$ for all $(x,y)\in A_j$. For each subset, two characteristics are of interest: the contact surface $A_j\cap B_i$ which is considered fixed and not random, and the number of RACs.

Let $N_{ij}= N_i(A_{j}\cap B_i)$ be the number of observable RACs on shoe $i$ and subset $j$, and let $n_{ij}$ denote its realization. Also, denote by $n_i=\sum_j n_{ij}$, the total number of observed RACs on shoe $i$. All the modeling assumptions above reduce to
\begin{equation}\label{modelass}
N_{ij}|a_i \sim \mbox{Poisson}(\lambda_j S_{ij} a_i),
\end{equation}
where $S _{ij} = |B_i \cap A_j|$ is the area of the contact surface of shoe $i$ in subset $j$. Figure~5 in Supporting web materials~3 summarizes the model schematically on a lab print.

%

\section{Estimation using maximum resolution}
The maximum resolution is achieved when the regions are in fact pixels. In this case, $S_{ij}\in\{0,1\}$ and $n_{ij}\in\{0,1\}$. When RACs are created they may tear the shoe sole such that the location of the RAC appears to be on an area with no contact surface and thus the value of $S_{ij}$ is set to 1 in all cases where $n_{ij}=1$. In addition, the value of $n_{ij}$ is set to 1 in 38 cases where $n_{ij}=2$. Appearance of two RACs in the same pixel may be due to the way the data were pre-processed and the location was defined. More details regarding the complexities in defining the location and in using pixels are addressed in Section~\ref{est_large}. Note that RACs cannot be observed in areas with no contact surface, thus $n_{ij}=0$ whenever $S_{ij}=0$.
The number of pixels 
is relatively large (more than 10,000 per shoe, on average)  compared to the number of RACs (average of approximately 34), meaning that  the event $n_{ij}=1$ is rare.

\subsection{A naive estimator for $\lambda^{(0)}$} \label{intes}
By \eqref{modelass}, if $S_{ij}>0$ then
\begin{equation*}
E\left(\frac{N_{ij}}{S_{ij}}\right)=E\left(\frac{\E(N_{ij}|a_i)}{S_{ij}}\right)\\
=E\left(\frac{\lambda_j S_{ij}a_i}{S_{ij}}\right)\\
=\lambda_j E\left( a_i\right)\\
=\lambda_j.
\end{equation*}
Therefore, a natural (unbiased) estimator for $\lambda_j$ is
\begin{equation}\label{unbiased}
\hat{\lambda}_j=\frac{1}{|m_j|}\sum_{i\in m_j} \frac{n_{ij}}{S_{ij}},
\end{equation}
where $m_j=\{i|S_{ij}>0\}$. As noted above, in case of pixels, $S_{ij}\in \{0,1\}$, $|m_j|=\sum_{i=1}^m S_{ij}$, and the estimator reduces to
\begin{equation*}
\hat{\lambda}_j=\frac{1}{|m_j|}\sum_{i\in m_j} n_{ij}
= \frac{\sum_{i=1}^m n_{ij}}{\sum_{i=1}^m S_{ij}},
\end{equation*}
where the last equality follows from $\{S_{ij}=0\} \Rightarrow \{n_{ij}=0\}$.

For the maximal resolution case, the estimator coincides with the one suggested by Wiesner et al. (2019).
In order to get a smooth estimator for the two-dimensional function $\lambda^{(0)}$, a kernel smoother is applied to the set of estimators $\hat{\lambda}_j$, $j=1,\dots,J$. 
This approach estimates a very large number of parameters separately, ignoring the spatial structure. A possible alternative is to model this structure using smooth functions, as is done next.

\subsection{Estimation of $\lambda^{(0)}$ using a random effects model} \label{Random}
In the binary setting, the occurrence of RACs is approximated by a logistic regression model. 
Specifically, for $S_{ij}=1$, it is assumed that
\begin{equation}\label{logist}
\P(N_{ij}=1 \mid a_i)= e^{g(\beta,x(A_j),y(A_j)) +a_i}/\{1+e^{g(\beta,x(A_j),y(A_j)) +a_i}\}
\end{equation}
where $(x(A_j),y(A_j))$ are the coordinates of pixel $j$ and  $g$ is chosen here to be a product of natural cubic splines $g(\beta,x,y)=g_X(x)g_Y(y)$ (Hastie et al., 2009).
%


A random effects model assumes that $a_i$ ($i=1,\ldots,m$) are independent and identically distributed having a law $H_\theta$ (Myers et al.,~2012~ p.319) indexed by a parameter $\theta\in \Theta$. Here the standard assumption that $H_\theta$ is a normal distribution with zero mean and unknown variance is used. The likelihood reduces to
\begin{equation*}
L(g(\beta,x(A_j),y(A_j)),\theta;n_{ij})=\Pi_{i=1}^m \int\Pi_{j=1}^J \left(\frac{e^{n_{ij}\cdot(g(\beta,x(A_j),y(A_j)) +a_i)}}{1+e^{g(\beta,x(A_j),y(A_j)) +a_i}}\right)^{S_{ij}} h_{\theta}(a)da,
\end{equation*}
where $h_\theta$ is the $N(0,\theta^2)$ density, and the estimators are obtained by maximizing the likelihood with respect to the parameters $\beta$ (of $g$) and $\theta$ (of $h$).
Since the number of pixels is relatively large (there are millions of binary variables having $S_{ij}=1$), computation is challenging and a  sub-sampling technique is used, as described in Section \ref{sampl}.

\subsection{Estimation of $\lambda^{(0)}$ using conditional maximum likelihood} \label{CML}
Instead of modeling the distribution of the shoe-specific parameter $a_i$, it can be treated as a nuisance parameter and be eliminated by conditioning on its sufficient statistic (Bishop et al., 2007; Agresti, 2013), leading to a conditional maximum likelihood (CML) approach.

The sufficient statistic of $a_i$ is  $n_i=\sum_j n_{ij}$, and the resulting conditional likelihood is,
\begin{equation*}
\frac{e^{\sum_{j=1}^J n_{\cdot j}\cdot g(\beta,x(A_j),y(A_j))}}{\Pi_{i=1}^m \sum_{u|n_{i}}  e^{\sum_{j=1}^J u_{j} \cdot g(\beta,x(A_j),y(A_j))}},
\end{equation*}
where $ n_{\cdot j}=\sum_i n_{ij}$  and $u|n_i$ indicates summation over all $u=(u_{1},\ldots,u_{J})$ such that $\sum_{j=1}^J u_{j}=n_i$.
This sum includes ${|B_i|}\choose{n_i}$ $ \approx$ $ {|B_i|}\choose{34}$ elements, where $|B_i|$ is very large (about $10,000$); see 
Figure~2 in Supporting web materials~2. This causes computational challenges and a case-control sub-sampling technique is used as described in the next section.

\subsection{Sub-sampling techniques} \label{sampl}

Estimating the intensity function at a high resolution is computationally challenging since the average number of pixels with contact surface per shoe is about 10,570, while the average number of RACs per shoe is around 34.  A possible approach is to use random sub-sampling for inference and specifically to employ case-control sub-sampling techniques, as recently suggested by Wright et al. (2017); (see also Fithian and Hastie, 2014).

In logistic regression, the estimated effect is consistent when using case-control sampling, but the estimated intercept may not be valid (Agresti, 2013). As we are mainly interested in $\lambda_0$, the intercepts are of a secondary importance, though they can be readily estimated as described by Wright et al. (2017).
The implication of random sub-sampling from the original data is as follows. Let $Z_{ij}$ indicate whether pixel $A_j$ of shoe $i$ is sampled ($1=$yes, $0=$no) and let $\rho^{(1)}_{i}=\P(Z_{ij}=1|N_{ij}=1)$ and $\rho^{(0)}_{i}=\P(Z_{ij}=1|N_{ij}=0)$ be the (possibly shoe-dependent) sub-sampling probabilities of cases and controls, respectively. It follows from Bayes' theorem that,
\begin{equation} \label{bays_ran}
\P(N_{ij}=1|Z_{ij}=1)=\frac{\rho^{(1)}_{i} \P(N_{ij}=1)}{\sum_{k=0}^{1} \rho^{(k)}_{i} \P(N_{ij}=k) }.
\end{equation}
\subsubsection{Case-control sub-sampling in a random effects model}\label{samp_rand}
Using \eqref{logist}, Equation~\eqref{bays_ran} simplifies to

\begin{equation*} \label{bays_ran1}
\P(N_{ij}=1|Z_{ij}=1)=\frac{\rho^{(1)}_{i} \exp(g(\beta,x(A_j),y(A_j)) +a_i)}{\rho^{(0)}_{i}+\rho^{(1)}_{i} (\exp(g(\beta,x(A_j),y(A_j)) +a_i) }
= \frac{ \exp(g^*(\beta,x(A_j),y(A_j)) +a_i)}{1+ \exp(g^*(\beta,x(A_j),y(A_j)) +a_i) },
\end{equation*}
where $g^*(\beta,x(A_j),y(A_j))=\log(\rho^{(1)}_{i}/\rho^{(0)}_{i})+g(\beta,x(A_j),y(A_j))$.


Thus, the underlying model of the random sub-sample is identical to the original model except for the intercept. Specifically, the location effect parameters are the same, and the intercept estimator is biased by a factor of $\log(\rho^{(1)}_{i}/\rho^{(0)}_{i})$. The latter can be adjusted with the inclusion of simple offset terms reflecting shoe-specific sampling probabilities (see Wright et al., 2017). Note that in the case of identical sampling probabilities for all shoes, the intercept can be adjusted after the estimation.

\subsubsection{Case-control sub-sampling in a CML estimation}\label{samp_CML}
Following similar steps as in the random effects case (see also Agresti, 2013, p.168)  it can be shown that the likelihood under case-control sub-sampling and the original likelihood differ only by the intercept. Since the intercept term does not appear in the conditional likelihood, the CML under both scenarios yields estimators for the same location effect parameters.


\subsection{Application to shoe data}\label{imple}

\subsubsection{Estimation of the baseline intensity function}
Figure~\ref{fig:comp_est} presents the three estimators applied to the shoe data.  The naive estimator of Section \ref{intes} was smoothed using the \textit{kernel2dsmooth} function in the R package smoothie (Gilleland, 2003). A uniform kernel is used where each entry of the smoothed matrix is calculated as the average of its $21^2$ neighbor entries in the original matrix.

The random effects and the CML estimates were calculated using a product of natural cubic splines. Three knots for the X-axis and five knots for the Y-axis were used and their positions were set according to equal quantiles. These numbers of knots enabled flexibility and still avoided computational problems. 

The calculations were based on within-cluster case-control sub-sampling that, as shown in section~\ref{comp_samp}, seems to outperform all other methods. From each shoe, the analysis includes all cases (pixels with RACs, $n_{ij}=1$) and a random sample of controls (pixels without RACs, $n_{ij}=0$) that is proportional to the number of cases in the shoe (20 controls for each case). If the number of controls on the shoe is less than required, all controls are used.

As is evident from Figure~\ref{fig:comp_est}, the CML and random effects estimate brought about similar results which were relatively close to the naive estimate (but much smoother). The estimated intensity function is highest at the ball and heel of the foot.
In addition, using the random effects model, the global test of all coefficients equal to zero (with the exception of the intercept) was conducted in order to check the constant intensity function assumption. Although the test was highly significant  in rejecting the constant intensity assumption ($p-value\approx 0$), the maximum estimated intensity value is about twice that of the minimum value; meaning that it is not far from a uniform intensity function. Thus, the probability of finding a RAC is relatively similar across the entire shoe sole, and there is no area in which observing a RAC increases dramatically the evidential value against a suspect.

\begin{figure}[tb]
	\begin{center}
		\includegraphics[width=0.8\textwidth]{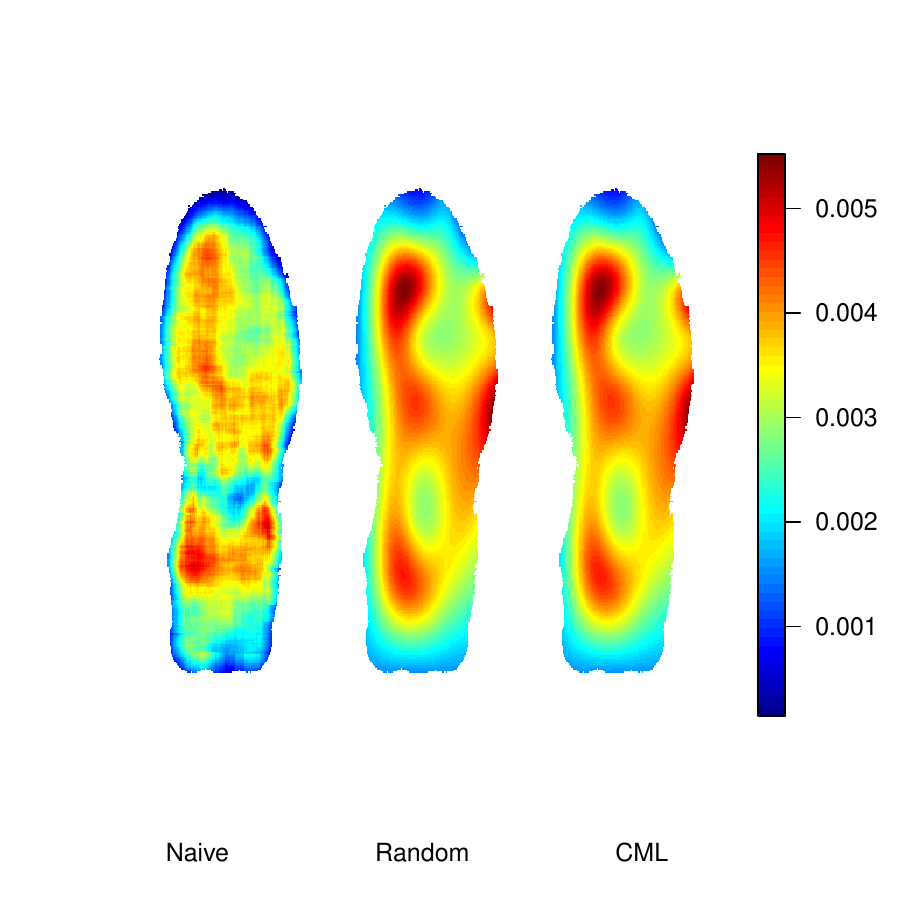}
		\\
	\end{center}
	\caption{Comparison of the three estimates of the intensity function}
	\label{fig:comp_est}
\end{figure}

\subsubsection{Confidence intervals}

To measure the uncertainty in the estimated intensity function, pointwise confidence intervals were calculated under the postulated model. A moment type estimator for the variance under the naive approach is given in Supporting web materials~4.
In the CML case, the asymptotic covariance matrix of $\hat{\beta}$ can be estimated using the observed information matrix based on conditional likelihood evaluated at $\hat{\beta}$ (Sartori and Severini, 2004). The covariance matrix was estimated using the \textit{clogit} function under the survival package (Therneau, 2015) in \R. In the random effects case, the same estimation approach of employing the observed information evaluated at $\hat{\beta}$ is performed using the random effects likelihood. The \textit{glmer} function under the lmr4 package (Bates et. al., 2015) in \R~was employed. Using these covariance matrices, confidence intervals of the intensity function were calculated. Figure~\ref{fig:comp_CI} presents the resulting 95\% pointwise confidence intervals of the estimators using random effects and CML models, in three chosen locations on the Y axis. The confidence intervals based on the naive estimator (not shown) are much wider as the result of the local estimation approach. The random and CML confidence intervals are relatively close and in most locations the CML is slightly wider. The deviation from the uniform model is quite clear. 
\begin{figure}[tb]
	\begin{center}
		\includegraphics[width=1.2\textwidth]{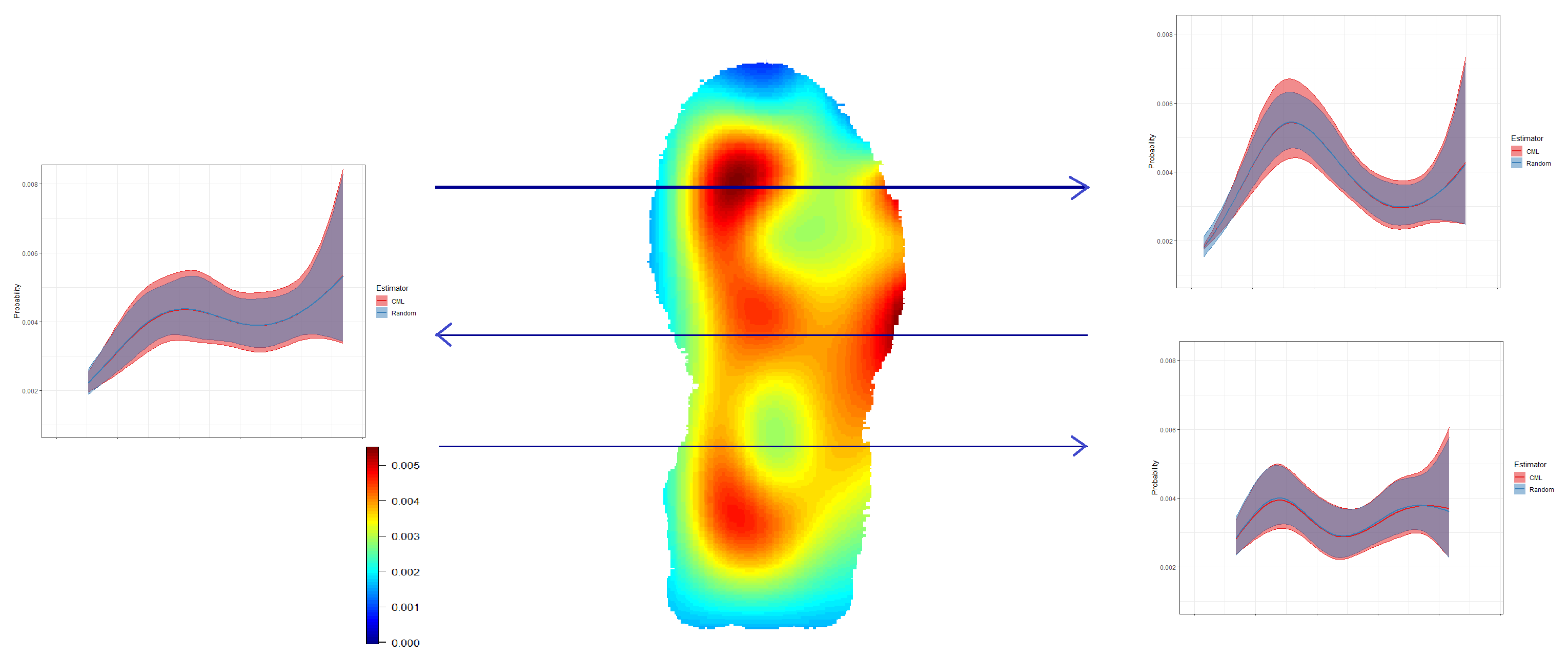}
		\\
	\end{center}
	\caption{Comparison of random effects and CML 95\% pointwise confidence intervals in three chosen locations}
	\label{fig:comp_CI}
\end{figure}


\section{Estimation using larger areas} \label{est_large}
\label{piece}
The estimates presented in the previous section are local and therefore heavily rely on valid location data. However, the definition of location is problematic for at least two reasons. First, a RAC is not a point in two dimensions but a set of points, which is marked by trained experts. The marking process is somewhat subjective and is exposed to marking errors, therefore the RACs' centers are prone to inaccuracies. Second, different shoes have different shapes, and it is not clear if they can be appropriately normalized. Here the shoes were normalized according to the Y axis, which is the standard measure of a size of a shoe, but the X axes of different shoes vary. This means that a RAC having a certain X-coordinate can appear near the middle in one shoe and near the edge in another.


Thus, ``locations'' of RAC's are more regional than local, and in order to overcome this, pixels should be grouped to larger  subsets. The question of how to divide the shoe and determine these subsets remains. As noted, these sets should be large enough in order to minimize the errors resulting from the normalization problem, especially on the X axis, but also should reflect the differences in the intensity function in different areas, which are mostly attributed to walking patterns. Based on the expertise of the authors from the police laboratory, it was decided that the Y axis of the shoe sole would be divided into 5 layers, the X axis would be divided into 2 layers, and the upper part of the shoe sole that comes in contact with the pad of the foot would be divided into an outer and inner part as they are expected to behave differently. Figure~\ref{fig:sub_exp} presents the resulting 14 areas that are believed to be quite homogeneous, but to have different probabilities of observing RACs. This partition is used in the following analysis.
\begin{figure}[!h]
	\begin{center}
		\includegraphics[width=0.5\textwidth]{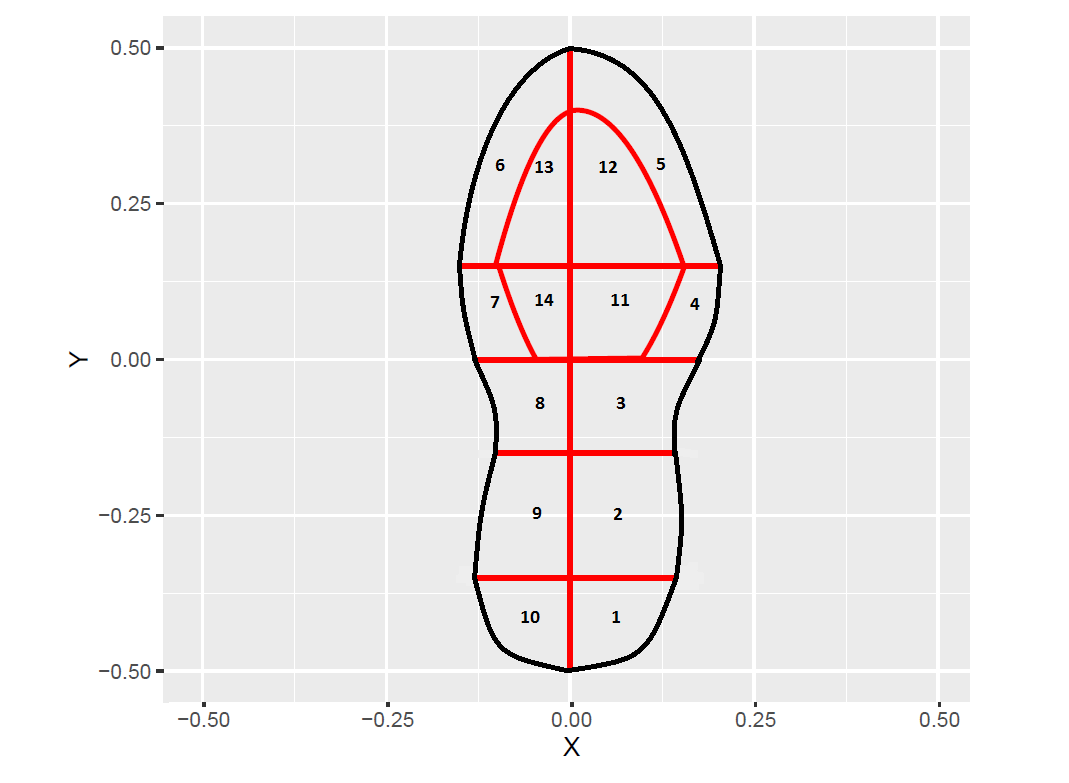}
		\\
	\end{center}
	\caption{Subsets of the shoe obtained according to expert knowledge}
	\label{fig:sub_exp}
\end{figure}

%

The number of RACs in each area follows the Poisson distribution given in \eqref{modelass}.
The naive approach results in the estimator \eqref{unbiased}. However, since the areas of subsets, $S_{ij}$, are not restricted to 0 or 1, this estimator does not coincide with the estimator presented in Wiesner et al. (2019).
%
%

For the random effects approach, the likelihood is
\begin{equation}
\label{ran_piece}
L(\lambda_1,\ldots,\lambda_J,\theta;n_{ij})=\Pi_{i=1}^m\int\Pi_{j=1}^J\frac{e^{-\lambda_{j}\cdot a \cdot S_{ij}} (\lambda_{j}\cdot a \cdot S_{ij})^{n_{ij}}}{n_{ij}!} h_{\theta}(a)da
\end{equation}
The random effects estimators are obtained by maximizing the likelihood with respect to  $\lambda_1,\ldots,\lambda_J$ and the parameters $\theta$ of $h_{\theta}(a)$.

The conditional maximum likelihood is obtained by conditioning on the $a_i$'s sufficient statistic, $N_i$. Since $N_{i1},...,N_{iJ}$ are independent, $N_{ij}|a_i\sim \mbox{Poisson}(\lambda_j a_i S _{ij})$ and $N_i|a_i\sim \mbox{Poisson}(\sum_j \lambda_j a_i S _{ij})$,
\begin{equation*} \label{cml_piece1_cont}
N_{i1}....N_{iJ}|N_i=n_i \sim \mbox{Multinomial}\left(n_i, \frac{e^{\log(\lambda_{j})+\log(S _{ij})}}{\sum_{j'} e^{\log(\lambda_{j'})+\log(S _{ij'})}}\right).
\end{equation*}
Thus, the log likelihood (up to a constant) is,
\begin{equation} \label{c}
\ell(\lambda_1,\dots,\lambda_J)=\sum_{i=1} ^m \sum_{j=1}^{J}n_{ij}\left[\log(\lambda_j)+\log(S _{ij})-\log\left(\sum_{j'} S _{ij'}\lambda_{j'}\right)\right].
\end{equation}
Since  $\ell(c\lambda_1,\dots,c\lambda_J)= \ell(\lambda_1,\dots,\lambda_J)$ for all $c>0$, the vector of parameters $(\lambda_1,\dots,\lambda_J)$ can be  estimated only up to a multiplicative constant. We therefore restrict $\lambda_1=1$, and the CML approach reduces to solving the following simple set of equations for  $\lambda_2,\ldots,\lambda_J$,
\begin{equation}
\label{cml_piece}
\frac{\partial \ell}{\partial \lambda_{k}} =\sum_{i=1}^m\left[\frac{n_{ik}}{\lambda_k}-\frac{n_i}{\sum_{j'} S_{ij'}\lambda_{j'} }S_{ik}\right]=0,
\end{equation}
which can be carried out numerically. This method is restricted to a relatively small number of regions, as the number of parameters $J$ must be small relative to the number of shoes $m$.

\begin{remark} \label{rem1}
Due to the identifiability issue discussed above, $\lambda_1,\ldots,\lambda_J$ cannot be fully estimated by using the conditional likelihood, which should be taken into account when comparing different methods. A possible approach is to scale the CML estimates  by equating their average to that of the estimated parameters under the naive approach \eqref{unbiased}. In this case, there is an added source of variance due to the variance of the naive estimators. This issue requires further research and is not addressed in the current study.	
\end{remark}

\subsection{The case of a single shoe model}

\begin{proposition} \label{samemod}
	
	If $S_{ij}=S_{j}, \forall i,j$, then $$\frac{\hat{\lambda}_j}{\hat{\lambda}_k}=\frac{n_{.j}}{n_{.k}}\cdot \frac{S_{k}}{S_{j}}$$ for any $k,j$ in all three estimators.
	
\end{proposition}
Proposition~\ref{samemod} refers to the simple case of all shoes having the same model, meaning that they have the same contact surface. Actually, it is true even if shoes have different models but have the same amount of contact surface in the different areas.
The Proposition states that the three methods differ only when the contact surface differs, 
and it holds for any partition of the shoe sole. The proof is deferred to Supporting web materials~5.

\subsection{Data analysis}
The results of the three estimators (naive, random and CML) applied to large areas are presented in Figure~\ref{fig:sub_areas}. In order to calculate the estimator in the random effects case, the \textit{hglm} function under the hglm package (R{\''o}nneg{\aa}rd, 2010) in \R~is used, where $a_i$ are $iid$  Gamma random variables.
 The two other estimators were implemented using a self written code. All three estimates agree on the areas with high and low intensity. The results are consistent with the previous analysis as the intensity function is the highest at the ball and heel of the foot. Here again the differences between the maximum and minimum estimated intensities are by a factor of 2.

In addition, 95\% confidence intervals based on the three aproaches were calculated. The interval of the naive estimator was calculated using the normal approximation with variance estimated as described in Section 4 of the Supporting web materials.  For the random effects estimator, the \textit{hglm} function under the hglm package (R{\''o}nneg{\aa}rd, 2010) was used to calculated the variance. The variance of the CML estimator is based on the observed information matrix.
The confidence intervals are presented in Table~\ref{table:1}. 

\begin{figure}[tb]
	\begin{center}
		\includegraphics[width=0.5\textwidth]{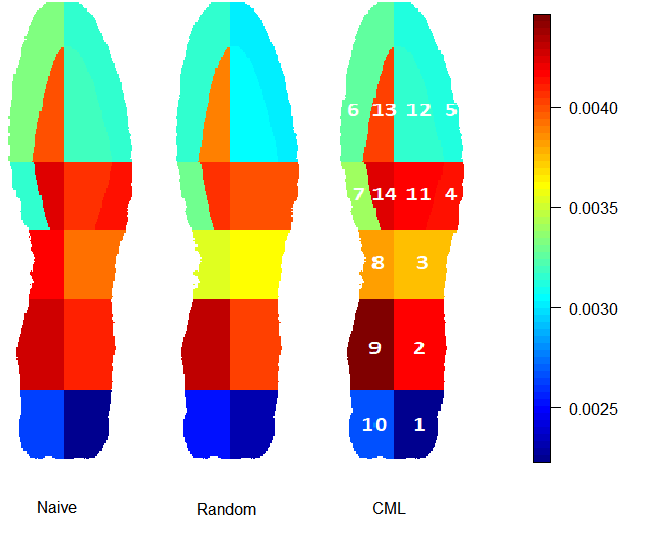}
	\end{center}
	\caption{Results of the estimators based on the piece-wise constant model.}
	\label{fig:sub_areas}
\end{figure}

\begin{table}[!h]
	\resizebox{10cm}{!}{
	\centering
	\begin{tabular}{ |p{3cm}|p{3cm}|p{3cm}|p{3cm}|  }
		\hline
		
		\hline
		Sub area& Naive & Random & CML\\
		\hline
		 1 & (0.002,0.0026) & (0.002,0.0026) & (0.0021,0.0024) \\
	2 & (0.0036,0.0046) & (0.0036,0.0045) & (0.004,0.0044) \\
	3 & (0.0033,0.0047) & (0.0031,0.0042) & (0.0035,0.0041) \\
	4 & (0.0036,0.0047) & (0.0035,0.0046) & (0.0038,0.0045) \\
	5 & (0.0028,0.0035) & (0.0027,0.0034) & (0.003,0.0033) \\
	6 & (0.003,0.0037) & (0.0028,0.0035) & (0.0031,0.0034) \\
	7 & (0.0026,0.0038) & (0.0028,0.0039) & (0.0031,0.0037) \\
	8 & (0.003,0.0057) & (0.0029,0.0043) & (0.0034,0.0043) \\
	9 & (0.0038,0.0048) & (0.0038,0.0048) & (0.0042,0.0047) \\
	10 & (0.0023,0.003) & (0.0022,0.0029) & (0.0025,0.0028) \\
	11 & (0.0036,0.0046) & (0.0035,0.0045) & (0.0039,0.0044) \\
	12 & (0.0028,0.0036) & (0.0027,0.0034) & (0.003,0.0033) \\
	13 & (0.0035,0.0045) & (0.0034,0.0044) & (0.0038,0.0043) \\
	14 & (0.0037,0.0049) & (0.0035,0.0047) & (0.0039,0.0046) \\
		
		\hline
	\end{tabular}
}\\
	\caption{Confidence intervals for the estimators based on the piece-wise constant model. The number of the sub-area indicates the area marked in Figure~\ref{fig:sub_areas}.}
	\label{table:1}
\end{table}

The confidence intervals of the three approaches are relatively close. The confidence interval of the CML approach is narrower as a result of the estimators' lower variance due to conditioning on the $n_i$'s and treating the scaling factor as a constant (see the discussion in Remark~\ref{rem1}).
The estimators agree on the areas with relatively wide and narrow intervals. The widest interval is of area 8 of the shoe, which is characterized by a low amount of contact surface (see Figure~3 in Supporting web materials~2).
In addition, using the random effects model, the hypothesis that the $\lambda_j$ parameters are equal for all $j$, meaning that the intensity is uniform over the whole shoe, is rejected with a $p-value\approx 0$.

%
%
%
%

\section{Simulation}\label{sim}
\subsection{Comparison of sub-sampling case-control techniques}\label{comp_samp}

Simulations are carried out to compare the random effects and CML estimators using different types of within-cluster case-control sub-sampling and sub-sampling across the whole data frame. Specifically, results based on the entire sample are compared to (i) random sub-sampling, (ii) case-control sub-sampling without taking into account the clusters, (iii) within-cluster case-control sub-sampling where controls  are sampled from each cluster proportional to the cluster size, and (iv) within-cluster case-control sub-sampling where controls  are sampled from each cluster  proportional to the number of cases in the cluster. In all scenarios, all cases are included in the sub-sample.

The full simulated data included 500 clusters with 500 observations in each. Events were drawn according to the logistic model:
\begin{equation*}\label{logist_sim}
\P(N_{ij}=1|a_i)= e^{\beta_0 +\beta_1\cdot x_{ij} + \beta_2\cdot x^2_{ij} +a_i}/(1+e^{\beta_0 +\beta_1\cdot x_{ij} + \beta_2\cdot x^2_{ij} +a_i}),
\end{equation*}
where $x_{ij}$ was evenly spaced between 0 and 1 (imitating a one dimensional location), and $\beta_0=-3, \beta_1=2, \beta_2=-2$.
The distribution of $a_i$ was $N(0,0.75^2)$.

The comparison was made based on 200 replications for each sampling technique. Figure~\ref{fig:coefsamp} 
shows the MSE of $\beta_1$ and $\beta_2$, where each sampling technique is marked with roman numerals presented above (0 indicates full data analysis without sub-sampling), ``r'' indicates that the results are based on the random effects estimator and  ``c'' indicates that they are based on the ``CML'' estimator. 

\begin{figure}[tb]
	
	\begin{tabular}{ll}
		
		\includegraphics[width=.6\textwidth]{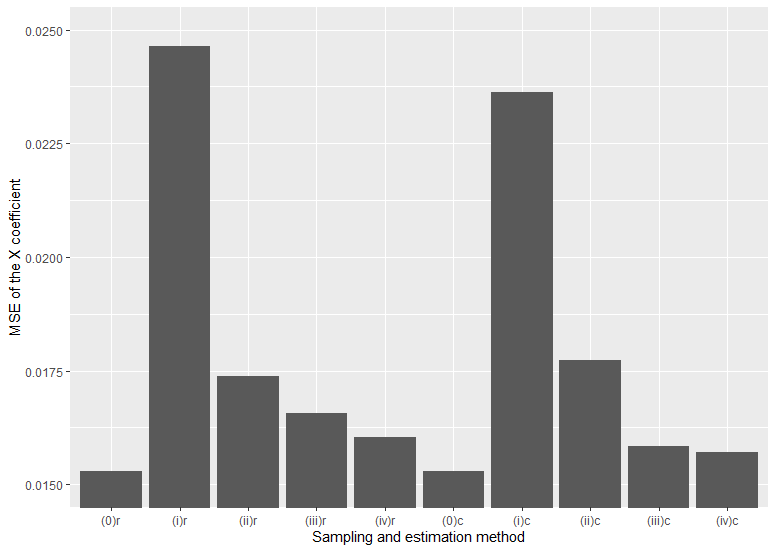}&
		\includegraphics[width=.6\textwidth]{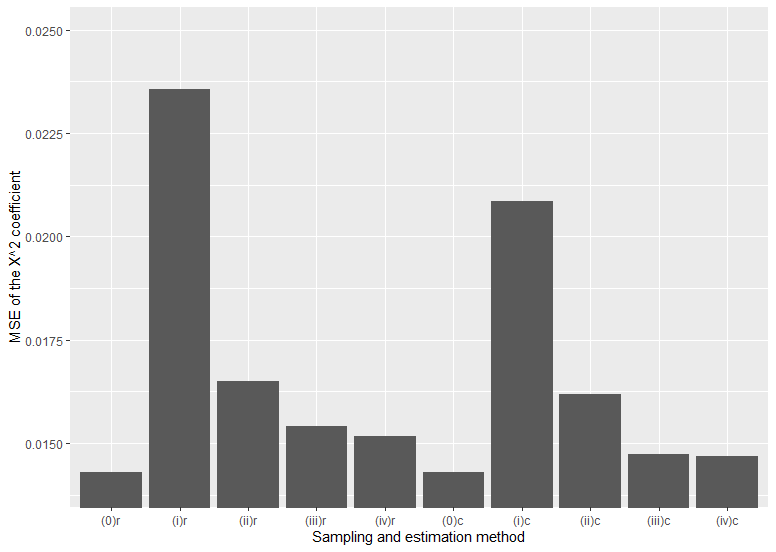}
	\end{tabular}
	
	\caption{MSE of $\beta_1$ (on the left) and $\beta_2$ (on the right) using different sub-sampling techniques.}
	\label{fig:coefsamp}
\end{figure}



Method 4 of within-cluster case-control sub-sampling where controls  are sampled from each cluster proportional to the number of cases in the cluster, 
performs best in terms of MSE. This is true for both the random effects and CML estimators, and for $\beta_1$ and $\beta_2$. In addition it can be seen that 
the random effects and CML estimators are relatively close in their MSE.

\subsection{Comparison of the three estimators}

A second simulation study compares the three estimators based on parameters from the shoe data, assuming 
that $N_{ij}|a_i\sim \mbox{Poisson}(a_i\lambda_j S_{ij})$. The first simulation uses the estimates for the $\lambda$'s obtained from the RACs' database by applying the naive approach and the observed contact surfaces, $S_{ij}$ in the  $J=14$ sub areas. Thus, the number of shoes in the simulation is 386, the same as in the original dataset. In each replication of the simulation, the  $a_i$ are simulated from a $\gamma(\theta,\theta)$ distribution, where $\theta=0.908$ is the estimate of $1/{\rm Var}(a)$ based on the data (see Supportig web materials 4). The results are based on 500 replications.

Figure~\ref{fig:lam1} 
presents the relative bias, which is the empirical bias of the estimator divided by the real parameters' value and the ratio between the estimators' MSE and the theoretical variance of the naive estimator (see Supporting web materials~4). The naive estimator is unbiased, and the simulation suggests that the bias of the two other estimators is negligible in this setting.  
The random effects estimator has the lowest relative bias and MSE across all parameters.  In addition, the naive estimators' MSE ratio is around 1 and since it is unbiased, this indicates that its theoretical variance is close to the empirical variance.

\begin{figure}[tb]
	\begin{tabular}{ll}
		
		\includegraphics[width=.6\textwidth]{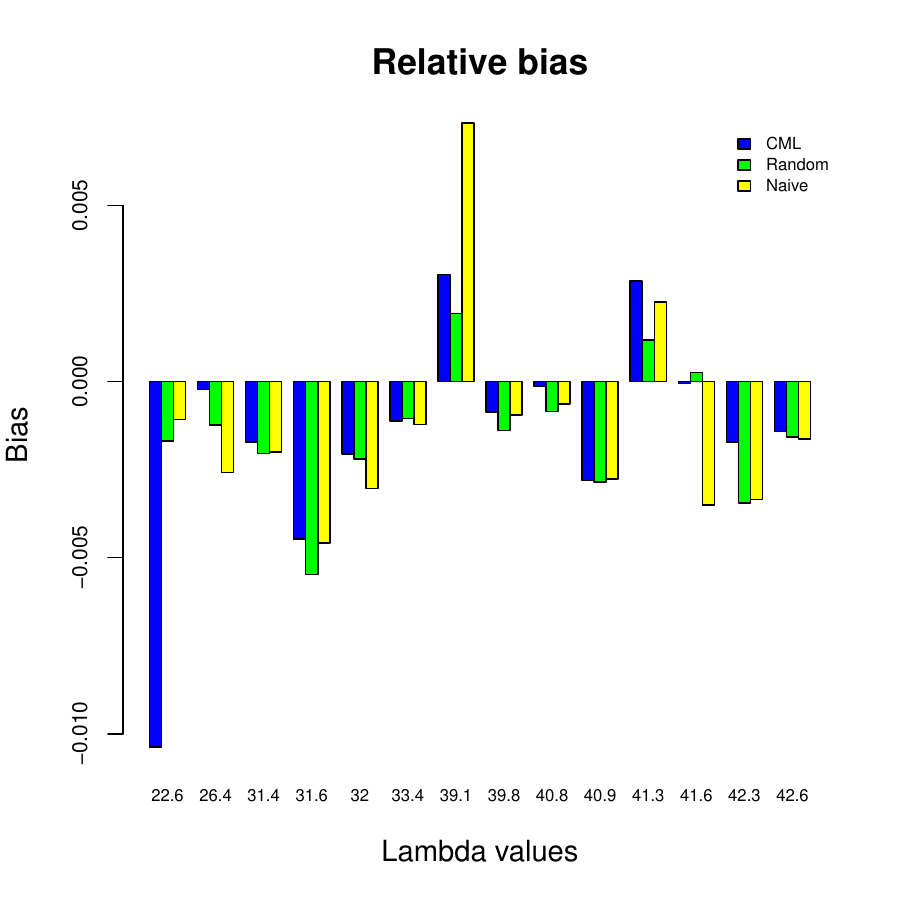}&
		\includegraphics[width=.6\textwidth]{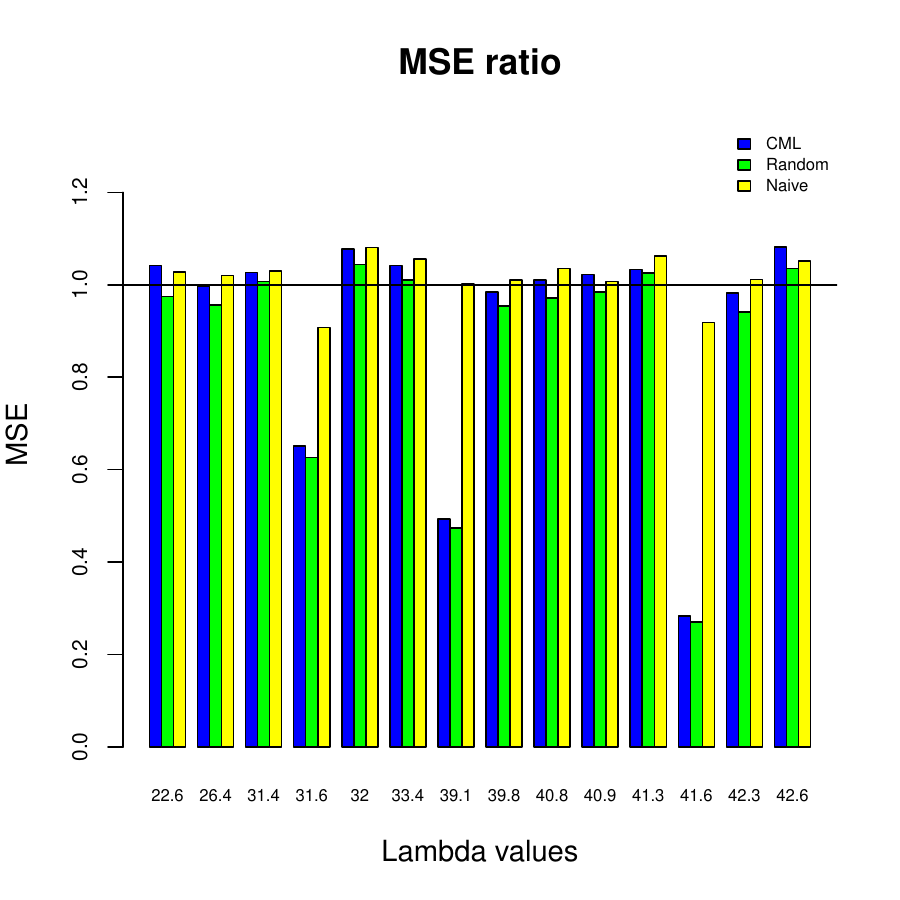}
	\end{tabular}
	\caption{The relative bias and MSE ratio}
	\label{fig:lam1}
\end{figure}


\flushleft{Additional simulations which investigate the effect of $\lambda$, the effect of the number of  sub areas, the effect of the sample size, the number of shoes and the effect of $a_i$ are conducted. 
	The specific settings and the results appear in an unpublished thesis (Kaplan-Damary, 2019). 	In summary, the random effects estimator is found to be the best among the estimators in most settings, with good performance in all. The CML estimator is very close to it.}

\section{Discussion}
\setlength{\parindent}{15pt}

The CML and random effects estimators produced very similar results, which were moreover relatively close to the naive estimator. The simulations suggest that, among the three estimators, the random effects perform best and thus may be preferred. Using this model, the hypothesis of a uniform intensity function is rejected, and the maximum estimated intensity value is approximately twice that of the minimum value.
The estimated intensity function is highest at the ball and heel of the foot. The deviation from uniformity is a result of the morphology of the foot and the areas of the foot that cause pressure on the shoe. This  fits the shape of the estimated intensity function presented here.

The findings of this study take us a step forward in assessing the potential evidential value of footprints.
It is clear that RACs in certain locations are rarer than in others, and these carry more weight in determining the rarity of a given shoe sole. Other characteristics of the RAC, such as size and shape, can be studied to further improve the classification of RACs as common or rare, but these are prone to severe measurement errors that must be modelled.  Additional challenges confront the researcher when analyzing crime scene prints, which are complicated by noise of various forms. Further research on actual crime scene RACs and their comparison to lab prints is therefore of vital importance.

The estimation of the intensity function has been made under the assumption of independence among RACs.
However, as recently shown by Kaplan Damary et al. (2018), this assumption is unjustified. This does not invalidate the findings, as using an independence working assumption results in consistent estimators, but the variance estimator may be somewhat biased. Using larger areas as building blocks for the model, may solve part of the local dependence problem, but further study is needed to understand the importance of the assumption.






\bibliography{bibfile}

\begin{thebibliography}{}

\bibitem[\protect\citename{Agresti, }2013]{agresti2013categorical}
Agresti, A. 2013.
\newblock {\em Categorical data analysis}. Third edn.
\newblock Hoboken, N.J.: John Wiley \& Sons.

\bibitem[\protect\citename{Bishop {\em et~al.}, }2007]{bishop2007estimating}
Bishop, Y.~M., Fienberg, S.~E., \& Holland, P.~W. 2007.
\newblock Estimating the size of a closed population.
\newblock {\em Discrete Multivariate Analysis: Theory and Practice},  229--256.

\bibitem[\protect\citename{Bodziak, }1999]{bodziak1999footwear}
Bodziak, W.~J. 1999.
\newblock {\em Footwear impression evidence: detection, recovery and
  examination}.
\newblock CRC Press.

\bibitem[\protect\citename{Cohen {\em et~al.}, }2011]{method2011Cohen}
Cohen, A., Wiesner, S., Grafit, A., \& Shor, Y. 2011.
\newblock A new method for casting three dimensional shoe prints and tire marks
  with dental stone.
\newblock {\em Journal of forensic sciences}, {\bf 56}(s1), S210--S213.

\bibitem[\protect\citename{Cressie, }1993]{cressie1993statistics}
Cressie, N. 1993.
\newblock {\em Statistics for Spatial Data: Wiley Series in Probability and
  Statistics}.
\newblock Wiley-Interscience New York.

\bibitem[\protect\citename{{ENFSI Expert Working Group Marks Conclusion Scale
  Committee and others}, }2006]{enfsi2006conclusion}
{ENFSI Expert Working Group Marks Conclusion Scale Committee and others}. 2006.
\newblock Conclusion scale for shoeprint and toolmarks examination.
\newblock {\em Journal of Forensic Identification}, {\bf 56}(2), 255--280.

\bibitem[\protect\citename{{Executive Office of the President President's
  Council of Advisors on Science and Technology}, }2016]{PCAST2016Forensic}
{Executive Office of the President President's Council of Advisors on Science
  and Technology}. 2016.
\newblock {\em {Forensic Science in Criminal Courts: Ensuring Scientific
  Validity of Feature-Comparison Methods}}.
\newblock Washington, D.C.: PCAST.

\bibitem[\protect\citename{Fithian \& Hastie, }2014]{fithian2014local}
Fithian, W., \& Hastie, T. 2014.
\newblock Local case-control sampling: Efficient subsampling in imbalanced data
  sets.
\newblock {\em Annals of statistics}, {\bf 42}(5), 1693--1724.

\bibitem[\protect\citename{Gilleland, }2013]{gilleland2013two}
Gilleland, E. 2013.
\newblock Two-dimensional kernel smoothing: Using the R package smoothie.
\newblock {\em NCAR Technical Note, TN-502+ STR}, {\bf 17}.

\bibitem[\protect\citename{Glattstein {\em et~al.}, }1996]{glattstein1996ph}
Glattstein, B., Shor, Y., Levin, N., \& Zeichner, A. 1996.
\newblock pH indicators as chemical reagents for the enhancement of footwear
  marks.
\newblock {\em Journal of Forensic Science}, {\bf 41}(1), 23--26.

\bibitem[\protect\citename{Hastie {\em et~al.}, }2009]{hastie2009unsupervised}
Hastie, T., Tibshirani, R., \& Friedman, J. 2009.
\newblock {\em The elements of statistical learning data mining, inference, and
  prediction}. Second edn.
\newblock New York, N.Y.: Springer.

\bibitem[\protect\citename{Hilderbrand, }2007]{hilderbrand2007footwear}
Hilderbrand, D.~S. 2007.
\newblock {\em Footwear: the missed evidence}. Second edn.
\newblock Staggs Pub.

\bibitem[\protect\citename{Kaplan-Damary, }2014]{Kaplan-Damary2014The}
Kaplan-Damary, N. 2014.
\newblock {\em {The Process of Documenting Shoe Prints, \\
  \textit{http://www.youtube.com/watch?v=sSLQ8G8qnGE$\&$feature=youtu.be}}}.

\bibitem[\protect\citename{Myers {\em et~al.}, }2012]{myers2012generalized}
Myers, R.~H., Montgomery, D.~C., Vining, G.~G., \& Robinson, T.~J. 2012.
\newblock {\em Generalized linear models: with applications in engineering and
  the sciences}.
\newblock  Vol. 791.
\newblock Hoboken, N.J.: John Wiley \& Sons.

\bibitem[\protect\citename{{National Research Council},
  }2009]{national2009strengthening}
{National Research Council}. 2009.
\newblock {\em {Strengthening forensic science in the United States: a path
  forward}}.
\newblock Washington, D.C.: The National Academies Press.

\bibitem[\protect\citename{{R Core Team}, }2014]{}
{R Core Team}. 2014.
\newblock {\em R: A Language and Environment for Statistical Computing}.
\newblock R Foundation for Statistical Computing, Vienna, Austria.

\bibitem[\protect\citename{Richetelli {\em et~al.},
  }2019]{richetelli2019empirically}
Richetelli, N., Bodziak, W.~J., \& Speir, J.~A. 2019.
\newblock Empirically observed and predicted estimates of chance association:
  Estimating the chance association of randomly acquired characteristics in
  footwear comparisons.
\newblock {\em Forensic science international}, {\bf 302}, 1--14.

\bibitem[\protect\citename{Sartori \& Severini, }2004]{sartori2004conditional}
Sartori, N., \& Severini, T.~A. 2004.
\newblock Conditional likelihood inference in generalized linear mixed models.
\newblock {\em Statistica Sinica},  349--360.

\bibitem[\protect\citename{Shor {\em et~al.}, }1998]{shor1998use}
Shor, Y., Vinokurov, A., \& Glattstein, B. 1998.
\newblock The use of an adhesive lifter and pH indicator for the removal and
  enhancement of shoeprints in dust.
\newblock {\em Journal of Forensic Science}, {\bf 43}(1), 182--184.

\bibitem[\protect\citename{Shor {\em et~al.}, }2003]{shor2003lifting}
Shor, Y., Tsach, T., Vinokurov, A., Glattstein, B., Landau, E., \& Levin, N.
  2003.
\newblock Lifting shoeprints using gelatin lifters and a hydraulic press.
\newblock {\em Journal of forensic sciences}, {\bf 48}(2), 368--372.

\bibitem[\protect\citename{Speir {\em et~al.}, }2016]{speir2016quantifying}
Speir, J.~A., Richetelli, N., Fagert, M., Hite, M., \& Bodziak, W.~J. 2016.
\newblock Quantifying randomly acquired characteristics on outsoles in terms of
  shape and position.
\newblock {\em Forensic Science International}, {\bf 266}, 399--411.

\bibitem[\protect\citename{Spencer \& Murray, }2019]{spencer2019bayesian}
Spencer, N.~A., \& Murray, J.~S. 2019.
\newblock A Bayesian Hierarchical Model for Evaluating Forensic Footwear
  Evidence.
\newblock {\em arXiv preprint arXiv:1906.05244}.

\bibitem[\protect\citename{Stone, }2006]{stone2006footwear}
Stone, R.~S. 2006.
\newblock Footwear examinations: mathematical probabilities of theoretical
  individual characteristics.
\newblock {\em Journal of Forensic Identification}, {\bf 56}(4), 577.

\bibitem[\protect\citename{SWGTREAD, }2005a]{SWGTREAD2005a}
SWGTREAD. 2005a.
\newblock Guide for the Detection of Footwear and Tire Impressions in the
  Field.
\newblock {\em J Forensic Ident}, {\bf 55}(6), 766--769.

\bibitem[\protect\citename{SWGTREAD, }2005b]{SWGTREAD2005b}
SWGTREAD. 2005b.
\newblock Guide for the Collection of Footwear and Tire Impressions in the
  Field.
\newblock {\em J Forensic Ident}, {\bf 55}(6), 770¬--773.

\bibitem[\protect\citename{SWGTREAD, }2005c]{SWGTREAD2005c}
SWGTREAD. 2005c.
\newblock Guide for the Preparation of Test Impressions from Footwear and
  Tires.
\newblock Available at
  \url{http://www.swgtread.org/images/documents/standards/published/swgtread_05_test_impressions_200503.pdf}.

\bibitem[\protect\citename{SWGTREAD, }2006a]{SWGTREAD2006a}
SWGTREAD. 2006a.
\newblock Guide for the Forensic Documentation and Photography of Footwear and
  Tire Impressions at the Crime Scene.
\newblock Available at
  \url{http://www.swgtread.org/images/documents/standards/published/swgtread_09_documentation_and_photography_200603.pdf}.

\bibitem[\protect\citename{SWGTREAD, }2006b]{SWGTREAD2006b}
SWGTREAD. 2006b.
\newblock Guide for the Examination of Footwear and Tire Impression Evidence.
\newblock Available at
  \url{http://www.swgtread.org/images/documents/standards/published/swgtread_08_test_examination_200603.pdf}.

\bibitem[\protect\citename{SWGTREAD, }2007a]{SWGTREAD2007a}
SWGTREAD. 2007a.
\newblock Guide for Casting Footwear and Tire Impression Evidence.
\newblock Available at
  \url{http://www.swgtread.org/images/documents/standards/published/swgtread_11_
  Casting Footwear and Tire Impression Evidence 200703.pdf}.

\bibitem[\protect\citename{SWGTREAD, }2007b]{SWGTREAD2007b}
SWGTREAD. 2007b.
\newblock Guide for Lifting Footwear and Tire Impression Evidence.
\newblock Available at
  \url{http://www.swgtread.org/images/documents/standards/published/swgtread_12_Lifting
  Footwear and Tire Impression Evidence 200703.pdf}.

\bibitem[\protect\citename{SWGTREAD, }2008]{SWGTREAD2008}
SWGTREAD. 2008.
\newblock Guide for Casework Documentation.
\newblock Available at
  \url{http://www.swgtread.org/images/documents/standards/published/swgtread_13_
  Casework Documentation 200803.pdf}.

\bibitem[\protect\citename{Therneau, }2015]{survival-package}
Therneau, T.M. 2015.
\newblock {\em A Package for Survival Analysis in S}.
\newblock version 2.38.

\bibitem[\protect\citename{Wiesner {\em et~al.}, }2019]{wiesner2019dataset}
Wiesner, S., Shor, Y., Tsach, T., Kaplan-Damary, N., \& Yekutieli, Y. 2019.
\newblock Dataset of Digitized RACs and Their Rarity Score Analysis for
  Strengthening Shoeprint Evidence.
\newblock {\em Journal of forensic sciences}, {\bf 65}(3), 762--774.

\bibitem[\protect\citename{Wright {\em et~al.}, }2017]{wright2017novel}
Wright, S.T., Ryan, L.M., \& Pham, T. 2017.
\newblock A novel case-control subsampling approach for rapid model exploration
  of large clustered binary data.
\newblock {\em Statistics in medicine}, {\bf 37}(6), 899--913.

\end{thebibliography}
\bibliographystyle{authordate1}
\nocite{*}

\newpage 
\appendix

\begin{center}
{\huge \textbf{Web-based supporting materials for "Spacial modeling for the probability of accidental mark locations on a shoe sole" }}
\end{center}

	\section{The process of evaluating shoeprints}\label{eval_shoe}
Here the process of documenting and evaluating shoeprints used in the Israel National Police
Division of Identification and Forensic Science (DIFS) is described.
\begin{enumerate}
	\item CSIs (Crime Scene Investigators) who arrive at the crime scene search for locations where the perpetrators have likely left their shoe prints, preferably an area left untouched by other inhabitants. Visible shoe prints are identified and the area is then darkened while investigators search for shoe prints using oblique light (SWGTREAD,2005a)
	\item After locating a shoeprint, the CSI places an L-shaped photography scale next to it, illuminates it at an angle that reveals as much detail as possible and photographs it with the camera positioned on a tripod directly over the print (SWGTREAD,2006)
	\item If possible, the surface on which the print appears is wrapped in paper and sent to the crime lab for further processing (SWGTREAD,2005b).
	\item	If this is not possible, prints are lifted using one of the following methods:
	\begin{enumerate}
		\item [a.] Two dimensional print methods (SWGTREAD,2007b and Manual for BVDA lifters); white adhesive lifter (most common method used in this lab for 2D shoeprints), black gelatin lifter and electrostatic lifter.
		\item[b.] Three dimensional Prints are casted using dental stone (SWGTREAD,2007a and Cohen et al., 2011).
	\end{enumerate}
	\item Once a suspect is apprehended, the suspect's shoes are sent to the crime lab for comparison with the shoeprints. Usually the shoes will be removed from the suspect's feet, but occasionally, the suspect's house will be searched for shoes with soles that resemble the shoeprints from the crime scene.
	\item All exhibits are registered by the investigating unit and are then sent to the crime lab through the ``evidence office'' which is responsible for assigning file numbers and passing on the exhibits to the relevant lab. The investigators add a letter that contains details about the crime and the collection of the relevant evidence.
	\item All the exhibits and their packaging are documented and marked at the crime lab. The examiner documents all of the information concerning the exhibits as they were received by the lab (date of reception, description of packaging and exhibit etc.) (SWGTREAD, 2008). This is known as the ``chain of custody.'' Keeping the chain of custody is important to ensure that the exhibits examined by the lab are the same exhibits confiscated in this case.
	\item All the shoeprints (exhibits, white adhesive lifters, black gelatin lifters and casts) and the shoes are photographed at high resolution (1000 dpi). The photographs taken at the crime scene are calibrated to 1000 dpi as well. Minimal image processing is carried out if necessary (conversion to gray scale, contrast and brightness adjustment etc.).
	\item If applicable, the shoeprints collected at the crime scene go through a process of enhancement. White adhesive lifters are sprayed with Bromophenol Blue which has a yellow color that turns blue when reacting with the dust, thus creating the shoeprint (Glattstein et al., 1996 and Shor et al., 1998). Shoeprints on items collected at the crime scene are lifted using the most appropriate method, black gelatin with a press (Shor et al., 2003) or white adhesive lifter.
	\item If the patterns are similar, as explained in the next paragraph, two dimensional lab prints are made from the suspect's shoes. The shoe soles are dusted with fingerprint powder and then impressions are made, while wearing the shoe, by stepping onto a clear adhesive film. At least two lab prints are made from each shoe (Hilderbrand, 2007), first by walking and again by pressing the adhesive to the shoe sole while it is in the air. If necessary, for the comparison stage additional three dimension lab  prints are made using Biofoam © (SWGTREAD, 2005c).
	
	\item Comparison and evaluation Stages (SWGTREAD, 2006b):
	The first step is visual examination of the crime scene prints and the shoe soles.
	\begin{enumerate}
		
		\item [a.]  Pattern - If the exact shoe pattern does not match, the result of the comparison is exclusion.
		\item [b.] Size - If the shoe pattern matches, but the physical size of shoeprint differs from the corresponding area of the shoe sole this might be explained by incorrect scaling, photography at an angle, movement of the shoe while producing the shoeprint etc. If no explanation satisfies the examiner as a logical explanation for the dissimilarities, exclusion is determined.
		\item [c.] Wear - If exclusion is not determined by now, the degree of general wear and local wear are compared. If the wear areas differ or crime scene print is worn more than the shoe, exclusion is determined. If the shoe is worn more than the crime scene print, the time elapsed between creation of the shoeprint and confiscation of the shoes is considered.
		\item [d.] The crime scene print is searched for locations where the pattern isn't complete and the reason might be the presence of RACs. If they indeed appear on both lab and crime scene prints, they are marked and their clarity, complexity and rarity are evaluated based on the examiner’s experience.
	\end{enumerate}
	
	\item If the patterns are similar, lab prints are made from the suspect's shoes as described above and these are compared to the shoeprints found at the crime scene. The two common comparison methods are:
	\begin{enumerate}
		\item [i.]  Overlay - a transparency of the suspect's shoe lab prints is positioned over a photograph of the print from the crime scene.
		\item [ii.] Side by Side - the lab prints and the photograph of the print from the crime scene are laid out side by side in order to compare the similarities between them.
		The method used in the lab is on screen overlay (using Lucia TrasoScan, by LIM ©) with 1000 ppi (pixel per inch) images.
		The examination process includes comparison of class characteristics (pattern, size and wear) and the identification of accidentals (RACs and unique wear)
		
	\end{enumerate}

	\item The examiner determines the level of certainty and writes a report. An ordinal scale is used to describe the degree of the match (ENFSI, 2006). The scale used in the Israel National Police Division of Identification and Forensic Science (DIFS) is presented next.
\end{enumerate}
\textbf{The conclusion scale:}
\begin{enumerate}
	\item [a.] Negative - the shoe under investigation is significantly different from the crime scene print, and therefore the suspect's shoe couldn't have left the crime scene print.
	\item [b.] Indication of non-association - Differences were found during the comparison but the quality of the print or the essences of the differences are not sufficient for total exclusion.
	\item [c.] Lacks sufficient details - The crime scene shoeprint lacks information that would enable significant comparison.
	\item [d.] Cannot be eliminated - Similarities in class characteristics were found, but the details on the crime scene print were limited and therefore, specific association is not possible.
	\item [e.] Possible - there is a match regarding class characteristics between the crime scene print and the suspect shoe.
	\item [f.]	Probable - Matching class and identifying characteristics were both found, but the amount of information in the identifying characteristics is not sufficient to determine Identification.		
	\item [g.] Highly probable – Similar to Probable, but a higher degree of certainty.
	
	\item [h.] Identification - Besides the match in class characteristics, the match in identifying characteristics is of sufficient quality and quantity.
\end{enumerate}

Based on the conclusion distribution of the lab, in approximately 40\% of all cases the suspects' shoes sent to the lab result in a non-match to the crime scene prints. A match of class characteristics (“Possible”) is found in approximately a third of the cases and accidentals are found in nearly 25\% of all cases. Less than 2\% of the crime scene prints lack sufficient details for comparison.

\section{Descriptive statistics} \label{app:0}

\begin{figure}[!h]
	\begin{center}
		\includegraphics[width=0.5\textwidth]{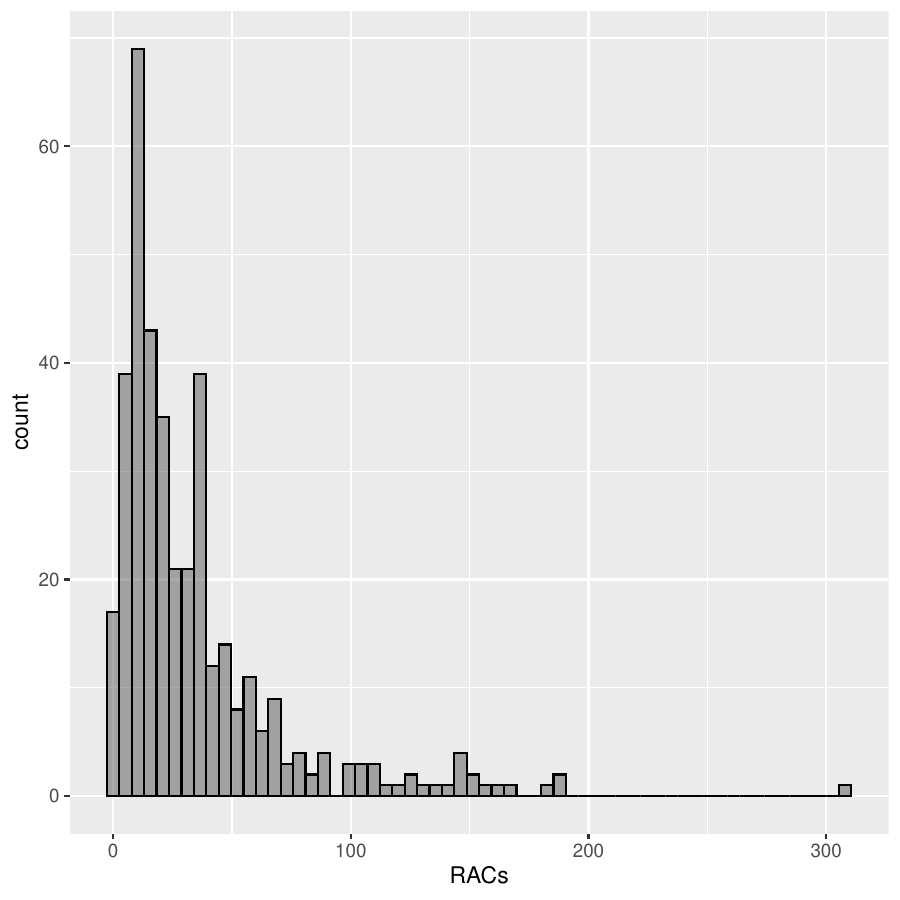}
	\end{center}
	\caption{A Histogram of the number of RACs per shoe}
	\label{fig:lab1}
\end{figure}

\begin{figure}[!h]
	\begin{center}
		\includegraphics[width=0.5\textwidth]{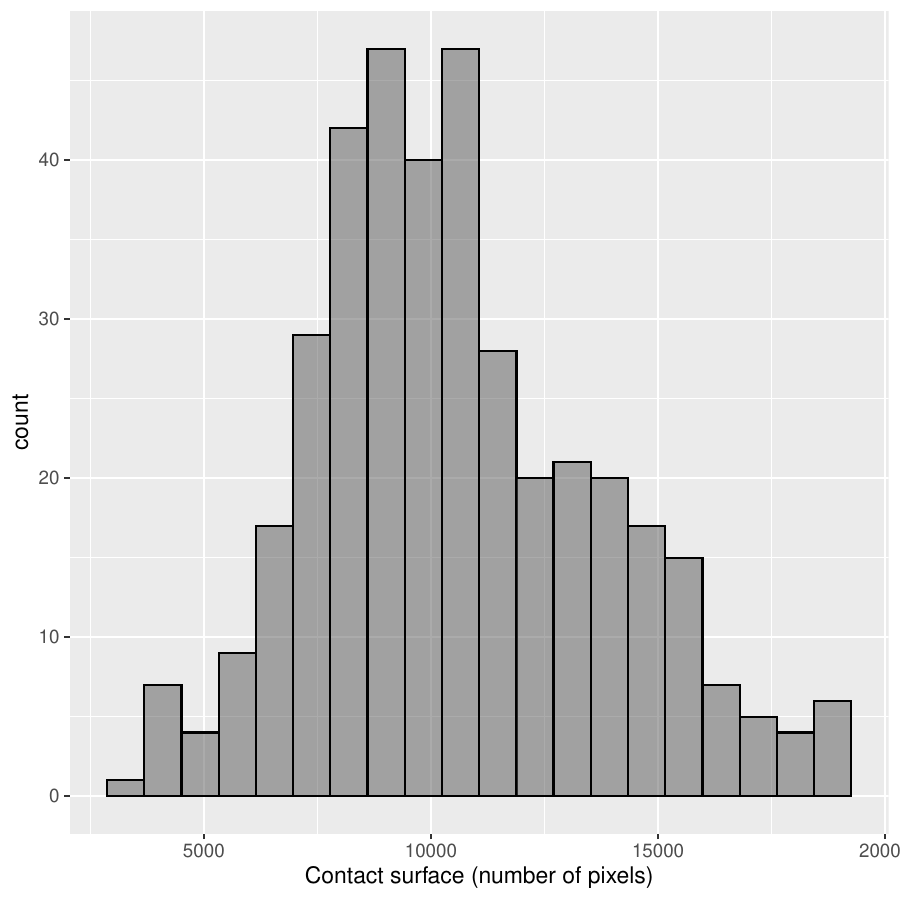}
	\end{center}
	\caption{A Histogram of the number of pixels with contact surface per shoe. An image of a shoe contains $395 \times 307 = 121,265$ pixels.}
	\label{fig:lab2}
\end{figure}

\begin{figure}[!h]
	\begin{center}
		\includegraphics[width=0.7\textwidth]{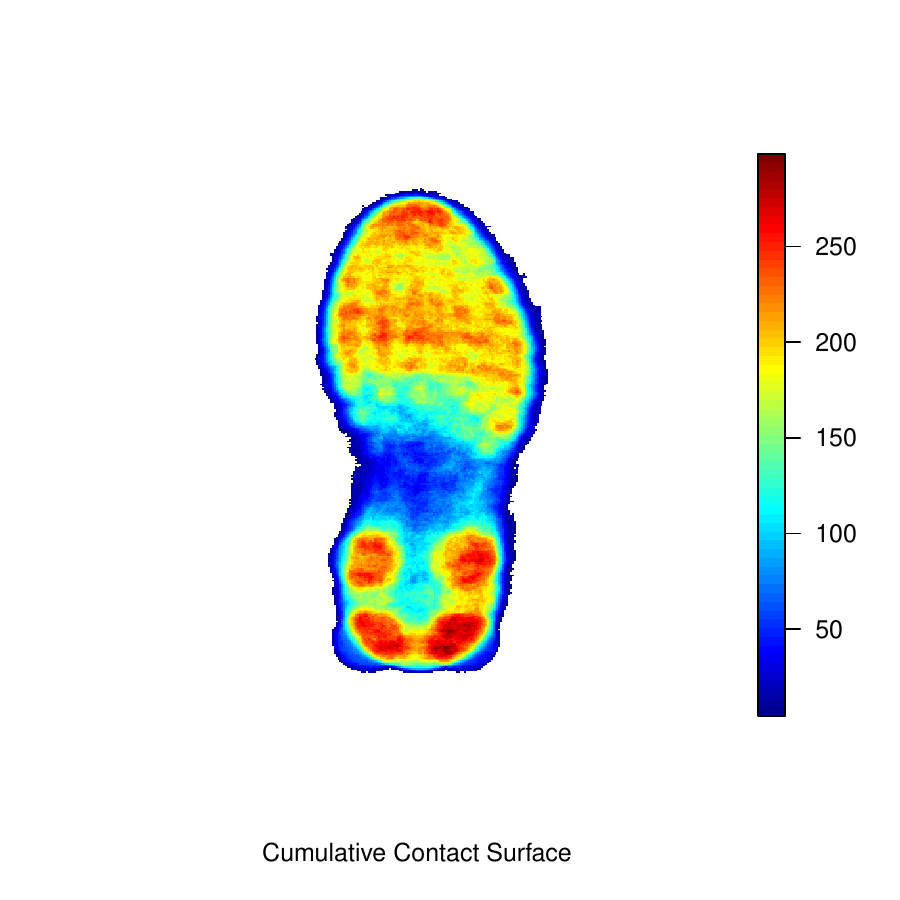}
	\end{center}
	\caption{A Cumulative Contact Surface of all shoes}
	\label{fig:comu}
\end{figure}

\begin{figure}[!h]
	\begin{center}
		\includegraphics[width=0.7\textwidth]{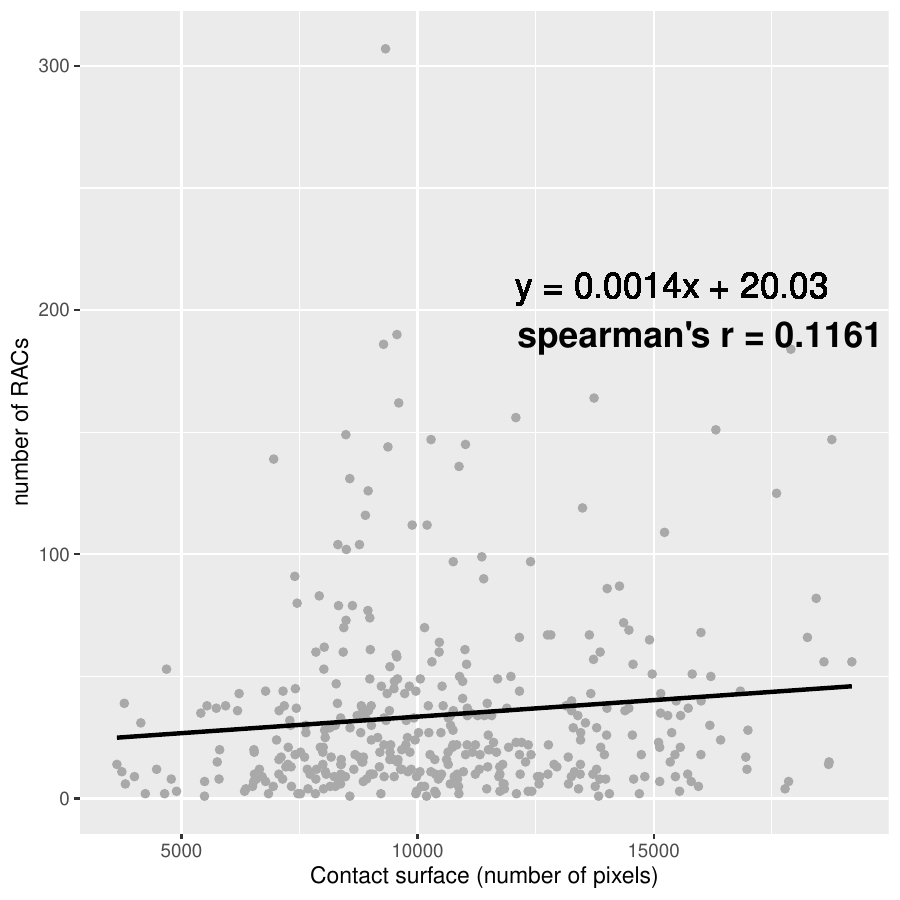}
	\end{center}
	\caption{A scatter plot of the contact surface and the number of RACs}
	\label{fig:lab3}
\end{figure}

\newpage

\section{Representation of the model parameters} \label{app:1}

Figure \ref{fig:par} presents the notation.	$\lambda_j$ is the intensity within pixel $j$, $a_i$ is the wear and tear parameter of shoe $i$, $S_{ij}$ is the area of the contact surface of shoe $i$ and pixel $j$ and $n_{ij}$ is the observed number of RACs on shoe $i$ and pixel $j$.

\begin{figure}[]
	\centering
	\begin{overpic}[width=0.175\textwidth]{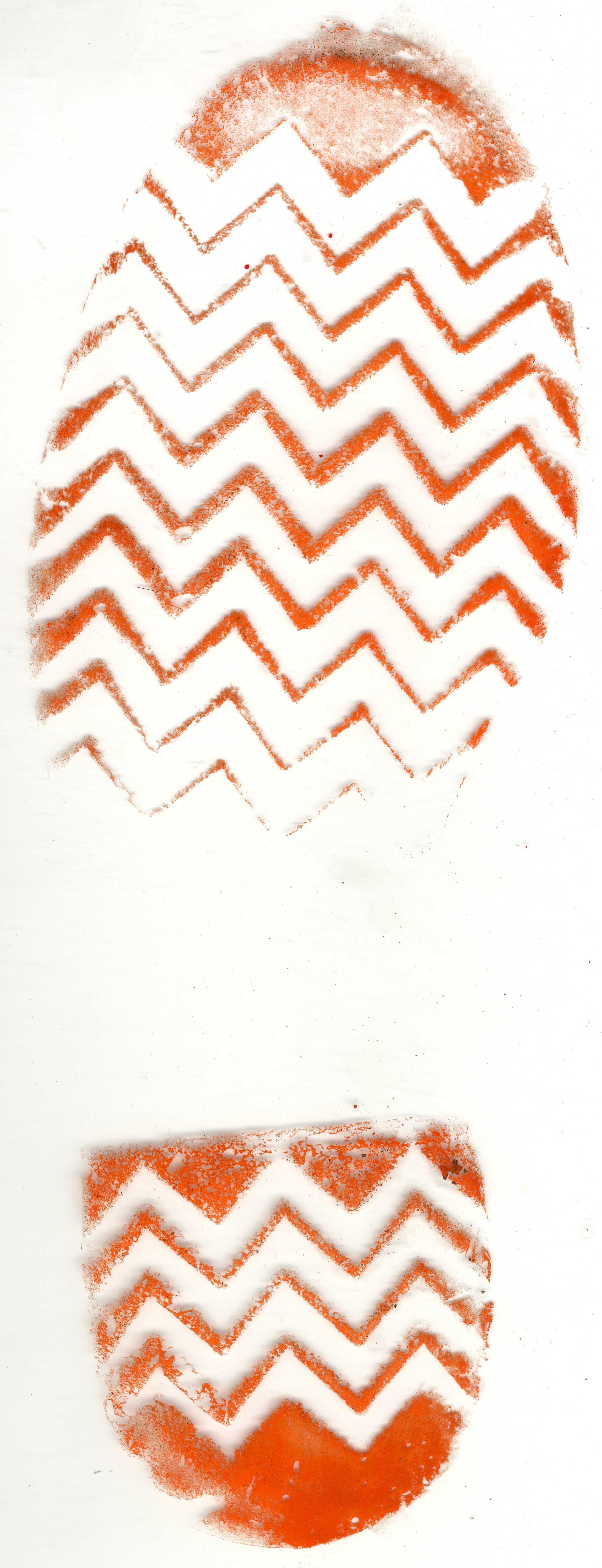}
		
		\put(-50,180){{\parbox{0.05\linewidth}{%
					\centerline{Shoe $i$}}}}
		

		\put(105,180){{\parbox{0.05\linewidth}{%
					\small{$$a_i , \lambda_j =? $$}}}}
		\put(105,165){{\parbox{0.05\linewidth}{%
					\tiny{$$i=1, \cdots ,M $$}}}}
		\put(105,160){{\parbox{0.05\linewidth}{%
					\tiny{$$j=1, \cdots, J$$}}}}
		
		\put(22,118){\vector(-4,1){45}}
		\put(-60,135){{\parbox{0.05\linewidth}{%
					\small{$$ n_{ir}=1$$}}}}
		
		\put(12,7){\vector(-4,1){30}}
		\put(-60,20){{\parbox{0.05\linewidth}{%
					\small{$$S_{ik}=0  $$}}}}
		\put(12,42){\vector(-4,1){28}}
		\put(-60,55){{\parbox{0.05\linewidth}{%
					\small{$$S_{i\ell}=0.5  $$}}}}

		\put(0,0){\line(1,0){75}}
		\put(0,5){\line(1,0){75}}
		\put(0,10){\line(1,0){75}}
		\put(0,15){\line(1,0){75}}
		\put(0,20){\line(1,0){75}}
		\put(0,25){\line(1,0){75}}
		\put(0,30){\line(1,0){75}}
		\put(0,35){\line(1,0){75}}
		\put(0,40){\line(1,0){75}}
		\put(0,45){\line(1,0){75}}
		\put(0,50){\line(1,0){75}}
		\put(0,55){\line(1,0){75}}
		\put(0,60){\line(1,0){75}}
		\put(0,65){\line(1,0){75}}
		\put(0,70){\line(1,0){75}}
		\put(0,75){\line(1,0){75}}
		\put(0,80){\line(1,0){75}}
		\put(0,85){\line(1,0){75}}
		\put(0,90){\line(1,0){75}}
		\put(0,95){\line(1,0){75}}
		\put(0,100){\line(1,0){75}}
		\put(0,105){\line(1,0){75}}
		\put(0,110){\line(1,0){75}}
		\put(0,115){\line(1,0){75}}
		\put(0,120){\line(1,0){75}}
		\put(0,125){\line(1,0){75}}
		\put(0,130){\line(1,0){75}}
		\put(0,135){\line(1,0){75}}
		\put(0,140){\line(1,0){75}}
		\put(0,145){\line(1,0){75}}
		\put(0,150){\line(1,0){75}}
		\put(0,155){\line(1,0){75}}
		\put(0,160){\line(1,0){75}}
		\put(0,165){\line(1,0){75}}
		\put(0,170){\line(1,0){75}}
		\put(0,175){\line(1,0){75}}
		\put(0,180){\line(1,0){75}}
		\put(0,185){\line(1,0){75}}
		\put(0,190){\line(1,0){75}}
		\put(0,195){\line(1,0){75}}
		
		\put(0,0){\line(0,1){195}}
		\put(5,0){\line(0,1){195}}
		\put(10,0){\line(0,1){195}}
		\put(15,0){\line(0,1){195}}
		\put(20,0){\line(0,1){195}}
		\put(25,0){\line(0,1){195}}
		\put(30,0){\line(0,1){195}}
		\put(35,0){\line(0,1){195}}
		\put(40,0){\line(0,1){195}}
		\put(45,0){\line(0,1){195}}
		\put(50,0){\line(0,1){195}}
		\put(55,0){\line(0,1){195}}
		\put(60,0){\line(0,1){195}}
		\put(65,0){\line(0,1){195}}
		\put(70,0){\line(0,1){195}}
		\put(75,0){\line(0,1){195}}

	\end{overpic}
	
	\captionof{figure}{Representation of the model parameters on the lab print}
	\label{fig:par}
\end{figure}

\newpage

\section{Variance of estimators under the naive approach}

Under the naive approach, $\hat{\lambda}_j=\sum_{i} N_{ij}/\sum_{i} S_{ij}$.
Since $N_{.j}|a_{i}\sim \mbox{Poisson}(\sum_{i}\lambda_{j}a_{i}S_{ij})$, where $N_{.j}=\sum_{i}N_{ij}$ and it is assumed that $\E(a_i)=1$, the variance of the estimator is
$$
\begin{array}{rl} \label{var_naive}
\var\left(\frac{1}{|m_j|}\sum_{i\in m_j} \frac{N_{ij}}{S_{ij}}\right)&=	\frac{1}{|m_j|^2}\sum_{i\in m_j}\frac{\var(N_{ij})}{S_{ij}^2}
=\frac{1}{|m_j|^2}\sum_{i\in m_j}\frac{\var\left(\E(N_{ij}|a_i)\right)+\E\left(\var(N_{ij}|a_i)\right)}{S_{ij}^2}\\ \\
&=\frac{1}{|m_j|^2}\sum_{i\in m_j}\frac{S_{ij}^2 \lambda_j^2 \var(a) +S_{ij} \lambda_j}{S_{ij}^2}
=\frac{\lambda_j^2 \var(a)}{|m_j|}+\frac{\lambda_j}{|m_j|^2}\sum_{i\in m_j}\frac{1  }{S_{ij}},
\end{array}
$$
which under maximal resolution ($S_{ij}=0,1$) reduces to
\begin{equation}\label{var_naive_pix}
\var(\hat{\lambda}_j)=\frac{\lambda_j^2 \var(a)+\lambda_j}{|m_j|}.
\end{equation}	

Let  $U_i=\frac{N_{i}^{2}-N_i}{(\sum_{j}\lambda_{j}S_{ij})^{2}}$, then simple calculations show
\begin{align*}
\E(U_i)&=
\frac{\E\left(\E^2(N_i|a_i)+\var(N_{i}|a_i)-\E(N_{i}|a_i)\right)}{(\sum_{j}\lambda_{j}S_{ij})^{2}}\\
& =\frac{\E\left((a_i\sum_{j}\lambda_{j}S_{ij})^2\right)}{(\sum_{j}\lambda_{j}S_{ij})^{2}}
=\var(a_i)+E^2(a_i)=\var(a)+1. \\
\end{align*}
Thus, by defining $\hat{U}_i=\frac{N_{i}^{2}-N_i}{(\sum_{j}\hat{\lambda}_{j}S_{ij})^{2}}$ with $\hat{\lambda}_j$ being the naive estimator, $\var(\hat{\lambda}_j)$ is readily estimated by plugging $\hat{\lambda}_j$ and $\widehat{\var}(a)=m^{-1}\sum_i\hat{U}_i-1$ in \eqref{var_naive_pix}.

\section{The case of a single shoe model}

\begin{proposition} \label{samemod}
	
	If $S_{ij}=S_{j}, \forall i,j$, then $$\frac{\hat{\lambda}_j}{\hat{\lambda}_k}=\frac{n_{.j}}{n_{.k}}\cdot \frac{S_{k}}{S_{j}}$$ for any $k,j$ in all three estimators.
	
\end{proposition}
\begin{proof}
	The naive estimator presented in 
	Section~4.1 reduces to:
	\begin{equation*}
	\hat{\lambda}_j=\frac{1}{m}\sum_{i} \frac{n_{ij}}{ S_{j}}=\frac{ n_{\cdot j}}{mS_{j}}
	\label{intuit2}
	\end{equation*}
	Note that in this case where the shoes are of the same shoe model, $m_j=m$.
	As mentioned in 
	Section~5, the CML estimator solves 
	equation (9) and in this case solves:
	\begin{equation*} \label{parlogcondl2}
	\sum_{i=1}^m\left[\frac{n_{ik}}{\hat{\lambda}_k}-\frac{n_i}{\sum_{j'} S_{j'}\hat{\lambda}_{j'} }S_{k}\right]=0.
	\end{equation*}
	It follows that
	\begin{equation*}
	\hat{\lambda}_k=\frac{n_{\cdot k}}{mS_{k}} \cdot\frac{m\sum_{j'}S_{j'} \lambda_{j'}}{n_{\cdot \cdot}} ,
	\label{solcml}
	\end{equation*}
	where $n_{\cdot \cdot}=\sum_{i=1}^m\sum_{j=1}^J n_{ij}$.
	The right fraction does not depend on $k$ and thus the ratio $\frac{\hat{\lambda}_j}{\hat{\lambda}_k}$ is equal to the ratio using the naive estimator.

	Using 
	equation (7) in 
	Section~5  the random effects log likelihood is proportional to the following equation (the $n_{ij}!$ is eliminated in the denominator)
	\begin{equation*}
	\label{l_ran_piece}
	\ell(\lambda_1,\ldots,\lambda_J,\theta;n_{ij})=\sum_{i=1}^m\log\int e^{-a\cdot\sum_{j=1}^J\lambda_{j}  \cdot S_{j}}\cdot a^{\sum_{j=1}^J n_{ij}} \Pi_{j=1}^J (S_{j}\lambda_J)^{n_{ij}} h_{\theta}(a)da
	\end{equation*}
	
	and the derivatives are
	\begin{equation}  \label{derrandnocont}
	\begin{aligned}
	\frac{\partial \ell}{\partial \lambda_k}&=\frac{\sum_{i=1}^m \frac{\partial}{\partial \lambda_k} \int e^{-a\sum_{j=1}^J S_{j} \lambda_{j} }\cdot a^{\sum_{j=1}^J n_{ij}}  \Pi_{j=1}^J (S_{j} \lambda_{j})  ^{n_{ij}} h_{\theta}(a)da}{\int e^{-a\sum_{j=1}^J S_{j} \lambda_{j} }\cdot a^{\sum_{j=1}^J n_{ij}}  \Pi_{j=1}^J ( S_{j} \lambda_{j})^{n_{ij}} h_{\theta}(a)da}\\
	&=-\sum_{i=1}^m\frac{ \int a \cdot S_{k}\cdot e^{-a\sum_{j=1}^J S_{j}\lambda_{j} }\cdot a^{\sum_{j=1}^J n_{ij}}  \Pi_{j=1}^J ( S_{j}\lambda_{j}) ^{n_{ij}} h_{\theta}(a)da}{\int e^{-a\sum_{j=1}^J S_{j}\lambda_{j} }\cdot a^{\sum_{j=1}^J n_{ij}}  \Pi_{j=1}^J (S_{j} \lambda_{j}) ^{n_{ij}} h_{\theta}(a)da}\\
	&+\sum_{i=1}^m \frac{\frac{n_{ik}}{\lambda_k} \int e^{-a\sum_{j=1}^J S_{j}\lambda_{j} }\cdot a^{\sum_{j=1}^J n_{ij}}  \Pi_{j=1}^J (S_{j} \lambda_{j}) ^{n_{ij}} h_{\theta}(a)da} {\int e^{-a\sum_{j=1}^J S_{j} \lambda_{j} }\cdot a^{\sum_{j=1}^J n_{ij}}  \Pi_{j=1}^J( S_{j} \lambda_{j})  ^{n_{ij}} h_{\theta}(a)da} \\
	&=-\sum_{i=1}^m\frac{\E(a \cdot S_{j}\cdot f_{i}(a))}{\E(f_{i}(a))}+\sum_{i=1}^m \frac{n_{ik}}{\lambda_k} \\
	&=-S_{k} \sum_{i=1}^m\frac{ \E(a \cdot f_{i}(a))}{\E(f_{i}(a))}+\sum_{i=1}^m \frac{n_{ik}}{\lambda_k}
	\end{aligned}
	\end{equation}
	where $f_{i}(a)=e^{-a\sum_{j=1}^J S_{j} \lambda_{j} }\cdot a^{\sum_{j=1}^J n_{ij}}  \Pi_{j=1}^J  (s_{j}\lambda_{j})  ^{n_{ij}}$.\\
	Note that $\frac{\E(a\cdot f_{i}(a))}{\E(f_{i}(a))}$ does not depend on $k$, i.e. it is equal for all $\lambda_k$ for any $h_{\theta}(a)$.  The random effects estimator is the solution of equating ~\eqref{derrandnocont} to 0. Denote the constant $\frac{\E(a\cdot f_{i}(a))}{\E(f_{i}(a))}$ by $c_i$ and $\sum_{i=1}^m n_{ik}$ by $n_{\cdot k}$.  Then, the estimator solves the following equation:
	\begin{equation*}
	-S_{k} \sum_{i=1}^m c_i+ \frac{n_{\cdot k}}{\hat{\lambda}_k}=0,
	\label{sol0randNoCont}
	\end{equation*}
	and thus, the random effects estimator is equal to:
	\begin{equation}
	\hat{\lambda}_k=\frac{n_{\cdot k}}{S_{k}\sum_{i=1}^m c_{i}}.
	\label{solrandNoCont}
	\end{equation}
	Thus, the ratio $\frac{\hat{\lambda}_j}{\hat{\lambda}_k}$ is equal for any $k,j$ in all three estimators.
\end{proof}
If $h_{\theta}(a)$ satisfies $c_i=\frac{\E(a\cdot f_{i}(a))}{\E(f_{i}(a))}=m$, 
all three estimators are identical.  An example of this is the case of $h_{\gamma}(a)=Gamma(\gamma,\gamma)$ where $\E(a)=1$ and $\var(a)=\frac {1}{\gamma}$. Here,
\begin{equation}
c_i=\frac{\E(a\cdot f_{i}(a))}{\E(f_{i}(a))}=\frac{n_{i\cdot}+\gamma}{\sum_{j=1}^J \lambda_j +\gamma}.
\label{ci}
\end{equation}
Summing Equation~\eqref{solrandNoCont} and placing $c_i$ of Equation~\eqref{ci} results in:
\begin{equation*}
\sum_{j=1}^J\hat{\lambda}_j=\frac{n_{\cdot \cdot}}{\frac{n_{\cdot \cdot}+m\gamma}{\sum_{j=1}^J\hat{\lambda}_j+\gamma}}.
\label{sumsolrandNoCont}
\end{equation*}
Thus, in this case:
\begin{equation}
\sum_{j=1}^J\hat{\lambda}_j=\frac{n_{\cdot \cdot}}{m}.
\label{sumressolrandNoCont}
\end{equation}
Now, summing Equation~\eqref{ci} and placing Equation~\eqref{sumressolrandNoCont} inside results in $\sum_{i=1}^m c_{i}=m$ and thus $\hat{\lambda}_j=\frac{ n_{\cdot j}}{m}$ and all three estimators are identical.

\end{document}